\documentclass[nopacs,nokeys,11pt,notitlepage,pra]{revtex4}

\usepackage{times}
\renewcommand{\baselinestretch}{1.0}
\usepackage{amsmath}
\usepackage[titletoc,title]{appendix}
\usepackage{graphicx,epic,eepic,epsfig,amsmath,latexsym,amssymb,verbatim,color}
\usepackage{dsfont}
\usepackage{float}
\usepackage{tikz}
\usepackage[strict]{changepage}
\usepackage{hyperref}
\hypersetup{colorlinks=true,citecolor=blue,linkcolor=blue,filecolor=blue,urlcolor=blue,breaklinks=true}

\usepackage[marginal]{footmisc}
\usepackage{url}
\usepackage{theorem}
\newtheorem{definition}{Definition}
\newtheorem{proposition}[definition]{Proposition}
\newtheorem{lemma}[definition]{Lemma}

\newtheorem{theorem}[definition]{Theorem}
\newtheorem{corollary}[definition]{Corollary}

\def\squareforqed{\hbox{\rlap{$\sqcap$}$\sqcup$}}
\def\qed{\ifmmode\squareforqed\else{\unskip\nobreak\hfil
\penalty50\hskip1em\null\nobreak\hfil\squareforqed
\parfillskip=0pt\finalhyphendemerits=0\endgraf}\fi}
\def\endenv{\ifmmode\;\else{\unskip\nobreak\hfil
\penalty50\hskip1em\null\nobreak\hfil\;
\parfillskip=0pt\finalhyphendemerits=0\endgraf}\fi}
\newenvironment{proof}{\noindent \textbf{{Proof~} }}{\hfill $\blacksquare$}
\newcounter{remark}
\newenvironment{remark}[1][]{\refstepcounter{remark}\par\medskip\noindent%
\textbf{Remark~\theremark #1 } }{\medskip}

\newcounter{example}

\mathchardef\ordinarycolon\mathcode`\:
\mathcode`\:=\string"8000
\def\vcentcolon{\mathrel{\mathop\ordinarycolon}}
\begingroup \catcode`\:=\active
  \lowercase{\endgroup
  \let :\vcentcolon
  }

\usepackage{cleveref}
\usepackage{graphicx}
\usepackage{xcolor}

\RequirePackage[framemethod=default]{mdframed}
\newmdenv[skipabove=7pt,
skipbelow=7pt,
backgroundcolor=darkblue!15,
innerleftmargin=5pt,
innerrightmargin=5pt,
innertopmargin=5pt,
leftmargin=0cm,
rightmargin=0cm,
innerbottommargin=5pt,
linewidth=1pt]{tBox}

\newmdenv[skipabove=7pt,
skipbelow=7pt,
backgroundcolor=blue2!25,
innerleftmargin=5pt,
innerrightmargin=5pt,
innertopmargin=5pt,
leftmargin=0cm,
rightmargin=0cm,
innerbottommargin=5pt,
linewidth=1pt]{dBox}
\newmdenv[skipabove=7pt,
skipbelow=7pt,
backgroundcolor=darkkblue!15,
innerleftmargin=5pt,
innerrightmargin=5pt,
innertopmargin=5pt,
leftmargin=0cm,
rightmargin=0cm,
innerbottommargin=5pt,
linewidth=1pt]{sBox}
\definecolor{darkblue}{RGB}{0,76,156}
\definecolor{darkkblue}{RGB}{0,0,153}
\definecolor{blue2}{RGB}{102,178,255}

\newcommand{\nc}{\newcommand}
\nc{\rnc}{\renewcommand}
\nc{\beg}{\begin{equation}}
\nc{\eeq}{{\end{equation}}}
\nc{\beqa}{\begin{eqnarray}}
\nc{\eeqa}{\end{eqnarray}}
\nc{\lbar}[1]{\overline{#1}}
\nc{\bra}[1]{\langle#1|}
\nc{\ket}[1]{|#1\rangle}
\nc{\ketbra}[2]{|#1\rangle\!\langle#2|}
\nc{\braket}[2]{\langle#1|#2\rangle}

\nc{\proj}[1]{| #1\rangle\!\langle #1 |}
\nc{\avg}[1]{\langle#1\rangle}
\nc{\Rank}{\operatorname{Rank}}
\nc{\smfrac}[2]{\mbox{$\frac{#1}{#2}$}}
\nc{\tr}{\operatorname{Tr}}
\nc{\ox}{\otimes}
\nc{\dg}{\dagger}
\nc{\dn}{\downarrow}
\nc{\cA}{{\cal A}}
\nc{\cB}{{\cal B}}
\nc{\cC}{{\cal C}}
\nc{\cD}{{\cal D}}
\nc{\cE}{{\cal E}}
\nc{\cF}{{\cal F}}
\nc{\cG}{{\cal G}}
\nc{\cH}{{\cal H}}
\nc{\cI}{{\cal I}}
\nc{\cJ}{{\cal J}}
\nc{\cK}{{\cal K}}
\nc{\cL}{{\cal L}}
\nc{\cM}{{\cal M}}
\nc{\cN}{{\cal N}}
\nc{\cO}{{\cal O}}
\nc{\cP}{{\cal P}}
\nc{\cQ}{{\cal Q}}
\nc{\cR}{{\cal R}}
\nc{\cS}{{\cal S}}
\nc{\cT}{{\cal T}}
\nc{\cV}{{\cal V}}
\nc{\cX}{{\cal X}}
\nc{\cY}{{\cal Y}}
\nc{\cZ}{{\cal Z}}
\nc{\cW}{{\cal W}}
\nc{\csupp}{{\operatorname{csupp}}}
\nc{\qsupp}{{\operatorname{qsupp}}}
\nc{\var}{{\operatorname{var}}}
\nc{\rar}{\rightarrow}
\nc{\lrar}{\longrightarrow}
\nc{\polylog}{{\operatorname{polylog}}}
\nc{\1}{{\mathds{1}}}
\nc{\wt}{{\operatorname{wt}}}
\nc{\av}[1]{{\left\langle {#1} \right\rangle}}
\nc{\supp}{{\operatorname{supp}}}

\def\a{\alpha}

\def\di{\diamondsuit}

\def\ve{\varepsilon}

\def\x{\xi}

\def\o{\omega}

\def\O{\Omega}
\def\w{\omega}

\nc{\RR}{{{\mathbb R}}}
\nc{\CC}{{{\mathbb C}}}
\nc{\FF}{{{\mathbb F}}}
\nc{\NN}{{{\mathbb N}}}
\nc{\ZZ}{{{\mathbb Z}}}
\nc{\PP}{{{\mathbb P}}}
\nc{\QQ}{{{\mathbb Q}}}
\nc{\UU}{{{\mathbb U}}}
\nc{\EE}{{{\mathbb E}}}
\nc{\id}{{\operatorname{id}}}

\nc{\CHSH}{{\operatorname{CHSH}}}

\nc{\be}{\begin{equation}}
\nc{\ee}{{\end{equation}}}
\nc{\bea}{\begin{eqnarray}}
\nc{\eea}{\end{eqnarray}}
\nc{\<}{\langle}
\rnc{\>}{\rangle}
\nc{\Hom}[2]{\mbox{Hom}(\CC^{#1},\CC^{#2})}
\nc{\rU}{\mbox{U}}

\nc{\ob}[1]{#1}

\nc{\SEP}{{\text{SEP}}}
\nc{\NS}{{\text{NS}}}
\nc{\LOCC}{{\text{LOCC}}}
\nc{\PPT}{{\text{PPT}}}
\nc{\EXT}{{\text{EXT}}}
\nc{\Sym}{{\operatorname{Sym}}}

\nc{\ERLO}{{E_{\text{r,LO}}}}
\nc{\ERLOCC}{{E_{\text{r,LOCC}}}}
\nc{\ERPPT}{{E_{\text{r,PPT}}}}
\nc{\ERLOCCinfty}{{E^{\infty}_{\text{r,LOCC}}}}
\nc{\Aram}{{\operatorname{\sf A}}}

\newcommand{\Choi}{Choi-Jamio\l{}kowski }

\usepackage{tikz}
\usepackage{hyperref}
\hypersetup{colorlinks=true,citecolor=blue,linkcolor=blue,filecolor=blue,urlcolor=blue,breaklinks=true}

\makeatletter
\def\grd@save@target#1{%
  \def\grd@target{#1}}
\def\grd@save@start#1{%
  \def\grd@start{#1}}
\tikzset{
  grid with coordinates/.style={
    to path={%
      \pgfextra{%
        \edef\grd@@target{(\tikztotarget)}%
        \tikz@scan@one@point\grd@save@target\grd@@target\relax
        \edef\grd@@start{(\tikztostart)}%
        \tikz@scan@one@point\grd@save@start\grd@@start\relax
        \draw[minor help lines,magenta] (\tikztostart) grid (\tikztotarget);
        \draw[major help lines] (\tikztostart) grid (\tikztotarget);
        \grd@start
        \pgfmathsetmacro{\grd@xa}{\the\pgf@x/1cm}
        \pgfmathsetmacro{\grd@ya}{\the\pgf@y/1cm}
        \grd@target
        \pgfmathsetmacro{\grd@xb}{\the\pgf@x/1cm}
        \pgfmathsetmacro{\grd@yb}{\the\pgf@y/1cm}
        \pgfmathsetmacro{\grd@xc}{\grd@xa + \pgfkeysvalueof{/tikz/grid with coordinates/major step}}
        \pgfmathsetmacro{\grd@yc}{\grd@ya + \pgfkeysvalueof{/tikz/grid with coordinates/major step}}
        \foreach \x in {\grd@xa,\grd@xc,...,\grd@xb}
        \node[anchor=north] at (\x,\grd@ya) {\pgfmathprintnumber{\x}};
        \foreach \y in {\grd@ya,\grd@yc,...,\grd@yb}
        \node[anchor=east] at (\grd@xa,\y) {\pgfmathprintnumber{\y}};
      }
    }
  },
  minor help lines/.style={
    help lines,
    step=\pgfkeysvalueof{/tikz/grid with coordinates/minor step}
  },
  major help lines/.style={
    help lines,
    line width=\pgfkeysvalueof{/tikz/grid with coordinates/major line width},
    step=\pgfkeysvalueof{/tikz/grid with coordinates/major step}
  },
  grid with coordinates/.cd,
  minor step/.initial=.2,
  major step/.initial=1,
  major line width/.initial=2pt,
}
\makeatother

\begin{document}

\title{Quantum Channel Simulation and the Channel's Smooth Max-Information}

\author{Kun Fang$^{1,2}$}
\email{kf383@cam.ac.uk}

\author{Xin Wang$^{1,3}$}

\author{Marco Tomamichel$^{1}$}

\author{Mario Berta$^{4}$}

\affiliation{$^1$Centre for Quantum Software and Information, Faculty of Engineering and Information Technology, University of Technology Sydney, NSW 2007, Australia}

\affiliation{$^2$ Department of Applied Mathematics and Theoretical Physics,\\ University of Cambridge, Cambridge, CB3 0WA, UK}

\affiliation{$^3$ Joint Center for Quantum Information and Computer Science, University of Maryland, USA}

\affiliation{$^4$ Department of Computing, Imperial College London, London, United Kingdom}

\thanks{A preliminary version of this paper was accepted as talk presentations at the 2018 IEEE International Symposium on Information Theory (ISIT 2018).}

\begin{abstract}
We study the general framework of quantum channel simulation, that is, the ability of a quantum channel to simulate another one using different classes of codes. First, we show that the minimum error of simulation and the one-shot quantum simulation cost under no-signalling assisted codes are given by semidefinite programs. Second, we introduce the channel's smooth max-information, which can be seen as a one-shot generalization of the mutual information of a quantum channel. We provide an exact operational interpretation of the channel's smooth max-information as the one-shot quantum simulation cost under no-signalling assisted codes, which significantly simplifies the study of channel simulation and provides insights and bounds for the case under entanglement-assisted codes. Third, we derive the asymptotic equipartition property of the channel's smooth max-information; i.e., it converges to the quantum mutual information of the channel in the independent and identically distributed asymptotic limit. This implies the quantum reverse Shannon theorem in the presence of no-signalling correlations. Finally, we explore the simulation cost of various quantum channels.
\end{abstract}

\maketitle


\section{Introduction}

Channel simulation is a fundamental problem in information theory. It asks how to use a (noisy) channel from Alice (A) to Bob (B) to simulate another (noisy) channel also from A to B~\cite{Kretschmann2004,Bennett2002}. Depending on the different resources available between A and B, this simulation problem has many variants.

For classical channels, Shannon's noisy channel coding theorem determines the capability of noisy classical channels to simulate noiseless ones~\cite{Shannon1948}. Dual to this famous coding theorem,  the `reverse Shannon theorem' concerns the use of noiseless channels to simulate noisy ones~\cite{Bennett2002}. Specifically, every channel can be simulated using an amount of classical communication equal to the capacity of the channel when there is free shared randomness between A and B in the asymptotic setting~\cite{Bennett2002}.
For quantum channels, the case when A and B share an unlimited amount of entanglement has been completely solved by the quantum reverse Shannon theorem (QRST)~\cite{Bennett2014,Berta2011a}, which states that the rate to optimally simulate a quantum channel in the asymptotic setting is determined by its entanglement-assisted classical capacity. In the zero-error scenario~\cite{Shannon1956}, using one channel to simulate another exactly with the aid of no-signalling correlations has been studied recently in~\cite{Cubitt2011,Duan2016,Wang2016b}, while the simulation with free quantum operations that completely preserve positivity of the partial transpose has been studied in~\cite{wang2018exact}. The problem of quantum channel simulations via other quantum resources has also been investigated in~\cite{Berta2013,BenDana2017}. 

In realistic settings, the number of channel uses is necessarily limited, and it is not easy to perform encoding and decoding circuits coherently over a large number of qubits in the near future. Therefore, it is important to characterize how well we can simulate a quantum channel from another with finite resources. The first step in this direction is to consider the one-shot setting. One-shot analysis has recently attracted great interest in classical information theory (see, e.g.,~\cite{Polyanskiy2010,Hayashi2009}) and quantum information theory (see, e.g.,~\cite{Tomamichel2012,Datta2013c,Buscemi2010b,Matthews2014,Wang2016g,Renes2011,Datta2013,Berta2016,Anshu2016b}). In one-shot information theory, the smooth max-information of a quantum state~\cite{Berta2011a} and its generalizations~\cite{Ciganovic2014} are all basic and useful quantities, which have various applications in quantum rate distortion theory as well as the physics of quantum many-body systems.

In this work, we focus on quantum channel simulation in both the one-shot and asymptotic regimes. The central quantity we introduce is the channel's smooth max-information. Our results can be summarized as follows. In Section~\ref{Channel simulation and codes}, we introduce the task of channel simulation and its related quantities. In Section~\ref{channel simulation part}, we develop a framework for quantum channel simulation assisted with different codes in the one-shot regime, where one has access only to a single use of the quantum channel. In particular, we characterize the minimum error of channel simulation under the so-called no-signalling (NS) codes \cite{Cubitt2011,Duan2016}, which allow the encoder and decoder to share non-local quantum correlations. Such codes are no-signalling from the sender to the receiver and vice versa, representing the ultimate limit of quantum codes obeying quantum mechanics and providing converse bounds for entanglement-assisted codes. The cost of approximately simulating a channel via noiseless quantum channels under NS-assisted codes can be characterized as an semidefinite program (SDP)~\cite{Vandenberghe1996}. In Section~\ref{max information part}, we introduce the channel's smooth max-information, which can be seen as a one-shot generalization of the mutual information of a quantum channel. Notably, this newly introduced entropy has the exact operational interpretation as the one-shot quantum simulation cost under NS-assisted codes. Then we prove its asymptotic equipartition property which directly implies the quantum reverse Shannon theorem in the presence of no-signalling correlations. 

In the setting of the entanglement-assisted one-shot capacity of quantum channels, Matthews and Wehner gave a converse bound in terms of the channel's hypothesis testing relative entropy~\cite{Matthews2014}. Moreover, a subset of us recently showed that the activated NS-assisted one-shot capacity is exactly given by the channel's hypothesis testing relative entropy~\cite{Wang2017a} -- generalizing the corresponding classical results~\cite{Polyanskiy2010,Matthews2012}. This suggests that the operational min- and max-type one-shot analogues of the channel's mutual information are the channel's hypothesis testing relative entropy and the channel's smooth max-information, respectively.

In Section~\ref{sec:examples}, as applications, we evaluate the cost of simulating fundamental quantum channels with finite resources. In particular, we derive a linear program to evaluate the finite blocklength simulation cost of quantum depolarizing channels.


\section{Channel simulation and codes}\label{Channel simulation and codes}

Let us now formally introduce the task of channel simulation and some basic notations. A quantum channel (quantum operation) $\cN_{A\to B}$ is a completely positive (CP) and trace-preserving (TP) linear map from operators acting on a finite-dimensional Hilbert space $\cH_{A}$ to operators acting on a finite-dimensional Hilbert space $\cH_{B}$. For any CPTP map $\cN_{A\to B}$, we will frequently use its \Choi matrix, which is defined as $J_{\cN}:= \sum_{i,j=0}^{|A|-1} \ket{i}\bra{j} \ox \cN_{A\to B} (\ket{i}\bra{j})$ where $\{\ket{i}\}$ are orthonormal basis in $\cH_A$. 

The general framework of quantum channel simulation is shown in Fig.~\ref{codes}. Here is how the simulation works. First, Alice performs some pre-processing via an encoder $\cE$ on the input system $A_i$ and outputs a quantum system $A_o$. Then she sends the output system to the shared quantum channel $\cN_{A_o\to B_i}$. At Bob's side, he receives a quantum system $B_i$ from the channel and performs his post-precessing via a decoder $\cD$. Finally, Bob outputs a quantum system $B_o$. The whole process can be considered as a quantum channel from Alice's input $A_i$ to Bob's output $B_o$, which we denote as $\widetilde M_{A_i\to B_o}$. The aim of simulation is to choose the best coding scheme to make the effective channel $\widetilde M_{A_i\to B_o}$ as similar as to the given target channel $\cM$. If we remove the channel $\cN$ from Fig.~\ref{codes}, we are left with a map with two inputs $A_i$, $B_i$ and two outputs $A_o$, $B_o$. We denote this map as $\Pi$, which generalizes the usual encoding scheme $\cE$ and decoding scheme $\cD$. Then the effective channel can be written as $\widetilde \cM_{A_i\to B_o}=\Pi_{A_iB_i\to A_oB_o}\circ\cN_{A_o\to B_i}$. Note that $\Pi$ is also known as a bipartite quantum supermap in some literatures~\cite{Chiribella2008,Gour2018,Gour2018a} since it maps a quantum channel to another quantum channel.

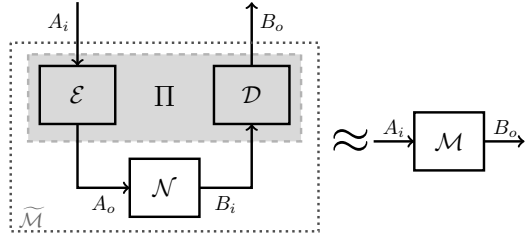
\begin{figure}[t]
\centering
\resizebox {7cm} {!} {
\begin{tikzpicture}[scale = 1]
    \def\xbb{0.6};\def\xb{1.5};\def\xsh{1.3};\def\ysh{1};
    \def\ya{1.3};\def\yc{3.4};\def\lo{0.2};\def\loo{0.4};
    \pgfmathsetmacro\xa{-\xb};\pgfmathsetmacro\xaa{-\xbb};
    \pgfmathsetmacro\xc{\xa-\xsh/2};\pgfmathsetmacro\xd{\xa+\xsh/2};
    \pgfmathsetmacro\xe{\xb-\xsh/2};\pgfmathsetmacro\xf{\xb+\xsh/2};
    \pgfmathsetmacro\yb{\ya+\ysh};
    \draw[very thick,dashed,fill=black!50,opacity=0.3] (\xc-\lo,\yb+\lo) rectangle (\xf+\lo,\ya-0.3);
    \draw[very thick,->] (\xa,\yc) -- node[left,shift={(0,0.2)}] {$A_i$} (\xa,\yb);
    \draw[very thick,<-] (\xb,\yc) -- node[right,shift={(0,0.2)}] {$B_o$} (\xb,\yb);
    \draw[very thick] (\xc,\yb) rectangle (\xd,\ya) node[midway] {\large $\cE$};
    \draw[very thick] (\xe,\yb) rectangle (\xf,\ya) node[midway] {\large $\cD$};
    \draw[very thick,->] (\xa,\ya) -- (\xa,0.2) -- node[below] {$A_o$} (\xaa,0.2);
    \draw[very thick,<-] (\xb,\ya) -- (\xb,0.2) -- node[below] {$B_i$} (\xbb,0.2);
    \draw[very thick] (\xaa,\ysh/2+0.2) rectangle (\xbb,-\ysh/2+0.2) node[midway] {\large $\cN$};
    \node[] (dots) at (0,\yb-0.52) {\Large $\Pi$};
    \draw[very thick,dotted,black!70] (\xc-\loo-0.1,\yb+\loo) rectangle (\xf+\loo+0.1,-\ysh/2-\lo+0.15);
    \node[black!60] (dots) at (\xc-0.15,-\ysh/2+0.2) {$\widetilde \cM$};
    \draw[very thick] (4.3,1.5) rectangle (5.5,0.5) node[midway] {\large $\cM$};
    \draw[very thick,->] (3.6,1) -- (4.3,1) node[above,shift={(-0.35,-0.06)}] {$A_i$};
    \draw[very thick,->] (5.5,1) -- (6.2,1) node[above,shift={(-0.33,-0.06)}] {$B_o$};
    \node[very thick] at (3.2,1) {\Huge $\approx$};
\end{tikzpicture}
}
\caption{General framework of quantum channel simulation.}
\label{codes}
\end{figure}

In particular, a bipartite quantum supermap $\Pi_{A_iB_i\to A_oB_o}$ is $A$ to $B$ no-signalling if $A$ cannot send classical information to $B$ by using $\Pi$. That is, for any quantum states $\rho_{1,A_i}$, $\rho_{2,A_i}$ and linear operator $\sigma_{B_i}$, we have $\tr_{A_o} \Pi(\rho_{1,A_i}\ox \sigma_{B_i}) = \tr_{A_o} \Pi(\rho_{2,A_i}\ox \sigma_{B_i})$, which is equivalent to the semidefinite condition~\cite{Duan2016} $\tr_{A_o}J_{\Pi}=\1_{A_i}/|A_i|\ox\tr_{A_oA_i}J_{\Pi}$ with $J_{\Pi}$ being the \Choi matrix of $\Pi$. Similarly, $B$ to $A$ no-signalling can be characterized by $\tr_{B_o}J_{\Pi}=\1_{B_i}/|B_i|\ox\tr_{B_iB_o}J_{\Pi}$. 
Note that the bipartite quantum supermap $\Pi$ in this work is required at least to be $B$ to $A$ no-signalling, which ensures the composition of $\Pi$ and $\cN$ leads to another quantum channel~\cite{Chiribella2008,Duan2016}. It is also worth mentioning that there is an isomorphism from bipartite quantum supermaps to bipartite quantum operations, and that the corresponding bipartite quantum operation is required to be no-signalling.

In the task of channel simulation, we say $\Pi$ is an $\O$-assisted code if it can be implemented by local operations with $\O$-assistance.  In the following, we eliminate $\O$ for the case of unassisted codes. We write $\O=\rm{NS}$ and $\O=\rm{PPT}$ for NS-assisted and positive-partial-transpose-preserving-assisted (PPT-assisted) codes, respectively~\cite{Leung2015c,Wang2017d}. In particular,
\begin{itemize}
	\item an unassisted code reduces to the product of encoder and decoder, i.e., $\Pi= \cD_{B_i\to B_o}\cE_{A_i\to A_o}$;
	\item a NS-assisted code corresponds to a bipartite quantum supermap which is no-signalling from Alice to Bob and vice-versa;
	\item a PPT-assisted code corresponds to a bipartite quantum supermap whose Choi-Jamio\l{}kowski matrix is positive under partial transpose over systems $B_iB_o$.
\end{itemize}

For any two quantum channels $\cN$ and $\cM$, the minimum error of simulation from $\cN$ to $\cM$ under $\O$-assisted codes is defined as 
\begin{align}\label{simulation}
\omega_\O(\cN,\cM):=\frac12 \inf_{\Pi \in \O}\|\Pi\circ\cN-\cM\|_\di,
\end{align}
where $\|\cF\|_\di:=\sup_{k\in \mathbb N} \sup_{\|X\|_1 \leq 1}\|(\cF\ox \id_k)(X)\|_1$ denotes the diamond norm and $\|X\|_1:=\tr \sqrt{X^\dagger X}$ denotes the Schatten~$1$-norm. The channel simulation rate from $\cN$ to $\cM$ under $\O$-assisted codes is defined as
\begin{align}
S_\O(\cN,\cM):=\lim_{\ve\to 0}\inf \left\{\frac{n}{m}\ \Big|\ \o_{\O}(\cN^{\ox n},\cM^{\ox m}) \leq \ve\right\},
\end{align}
where the infimum is taken over ratios $\frac{n}{m}$ with $n,m \in \mathbb N$. In this framework of channel simulation, the classical capacity $C(\cN)$ and the quantum capacity $Q(\cN)$ of the channel $\cN$ are respectively given by 
\begin{align}
\label{C Q simulation}
C(\cN)=S\big(\cN, \widehat{\mathrm{id}}_2\big)^{-1} \ \text{and} \ \ Q(\cN)=S(\cN, \id_2)^{-1},
\end{align}
where $\widehat{\id}_2$ is the one-bit noiseless channel and $\id_2$ is the one-qubit noiseless channel.

If we consider simulating the given channel $\cN$ via an $m$-dimensional noiseless quantum channel~$\id_m$, then the one-shot $\ve$-error quantum simulation cost under $\O$-assisted codes is defined as
\begin{align}\label{one shot cost definition}
S^{(1)}_{\Omega,\varepsilon}(\cN) := \log \inf \big\{m \in \mathbb N \ | \ \omega_\O(\id_m,\cN) \leq \varepsilon \big\},
\end{align}
where the logarithms in this work are taken in the base
two.
The asymptotic quantum simulation cost is given by the regularization
\begin{align}\label{eq: definition asymp cost}
S_\O(\cN) = \lim_{\ve \to 0} \lim_{n\to \infty} \frac{1}{n} S^{(1)}_{\Omega,\ve}(\cN^{\ox n}).
\end{align}


\section{Channel simulation via noisy quantum channels}\label{channel simulation part}

Based on the definitions in the above section, we show that the minimum error of simulation under NS-assisted (and PPT-assisted) codes can be given by SDPs. The one-shot $\ve$-error quantum simulation cost under NS-assisted codes can also be given by an SDP. These SDPs can be easily implemented for small blocklength and they also lay the foundation of analysis in the following sections. 

\begin{proposition}\label{condition of simulation}
	For any two quantum channels $\cN$ and $\cM$ with corresponding Choi-Jamio\l{}kowski matrices $J_\cN$ and $J_\cM$, the minimum error of simulation from $\cN$ to $\cM$ under NS-assisted codes $\omega_{\rm NS}(\cN,\cM)$ is given by the following SDP,
	\begin{subequations}\label{SDP S}
		\begin{align}
		\inf & \quad \ \gamma\\
		\text{\rm s.t.} &\quad  \tr_{B_o} Y_{A_iB_o} \leq \gamma \1_{A_i},\\
		& \quad \ Y_{A_iB_o}\ge  J_{\widetilde{\cM}}-J_\cM,\ Y_{A_iB_o} \geq 0, \\
		& \quad \ J_{\widetilde{\cM}} = \tr_{A_oB_i} (J_{\cN}^T\ox\1_{A_iB_o})J_{\Pi}, \label{condi 1}\\
		& \quad  \ J_{\Pi}\ge 0,\ \tr_{A_oB_o} J_{\Pi} = \1_{A_iB_i}, \hspace{1.55cm} \mathrm{(CP,TP)}\label{CPTP}\\
		&\quad \, \tr_{A_o}J_{\Pi}=\1_{A_i}/|A_i|\ox\tr_{A_oA_i}J_{\Pi},\label{A not to B} \hspace{0.8cm} (\mathrm{A} \not \to \mathrm{B})\\
		& \quad \, \tr_{B_o}J_{\Pi}=\1_{B_i}/|B_i|\ox\tr_{B_iB_o}J_{\Pi}.\hspace{0.8cm} (\mathrm{B} \not \to \mathrm{A})\label{B not to A}		
		\end{align}
	\end{subequations}
	To obtain $\o_{\rm{NS \cap PPT}}(\cN, \cM)$, we only need to add the PPT constraint $J_{\Pi}^{T_{B_iB_o}}\ge 0$, where $T_{B_iB_o}$ denotes the partial transpose over systems $B_iB_o$.
\end{proposition}

\begin{proof}	
	Note that for any two quantum channels $\cN_1, \cN_2$ from $A$ to $B$, the diamond norm of their difference can be expressed as an SDP of the following form~\cite{Watrous2009}: 		
	\begin{align}\label{SDP diamond}
		\frac12 \|\cN_1-\cN_2\|_\di= \inf \big\{\ \gamma \ \big|\ \tr_B Y \leq \gamma \1_A,\ Y \ge J_{\cN_1}-J_{\cN_2},\ Y \ge 0 \ \big\},
	\end{align}
	where $J_{\cN_1}$ and $J_{\cN_2}$ are the corresponding Choi-Jamio\l{}kowski matrices. We denote the Choi-Jamio\l{}kowski matrix of code $\Pi$ as $J_{\Pi}$. From Lemma~\ref{compose channel lemma} in the Appendix, we know that the Choi-Jamio\l{}kowski matrix of the effective channel $\widetilde \cM = \Pi\circ\cN$ is given by
	\begin{align}\label{code compose channel}
	J_{\widetilde{\cM}} = \tr_{A_oB_i} (J_{\cN}^T\ox\1_{A_iB_o})J_{\Pi}.
	\end{align}
	Together with the constraints of the code $\Pi$, we have the resulting SDP~\eqref{SDP S}. The constraints in Eq.~\eqref{CPTP} represent the CP and TP conditions of the bipartite supermap $\Pi$. The constraints in Eqs.~\eqref{A not to B} and~\eqref{B not to A} represent the no-signalling conditions that $A$ cannot signal to $B$ and $B$ cannot signal to $A$, respectively.
\end{proof}

\begin{corollary}\label{min error simulation}
	The minimum error to simulate a quantum channel $\cN$ via a noiseless quantum channel under NS-assisted codes $\omega_{\rm{NS}}(\id_m,\cN)$ is given by the following SDP, 
	\begin{subequations}\label{id error SDP}
	\begin{align}
	\inf & \quad \ \gamma\\
	\text{\rm s.t.} & \quad \tr_{B} Y_{AB} \leq \gamma \1_{A},\\
	& \quad \ Y_{AB}\ge J_{\widetilde{\cN}} - J_\cN,\ Y_{AB} \geq 0,\\
	& \quad \ J_{\widetilde{\cN}} \ge 0,\ \tr_{B} J_{\widetilde{\cN}} = \1_{A},\\
	&\quad \ J_{\widetilde{\cN}} \leq \1_{A} \otimes V_{B},\ \tr V_{B} = m^2.	
	\end{align}
	\end{subequations}
	To obtain $\o_{\rm{NS \cap PPT}}(\id_m,\cN)$, we only need to add the PPT constraint 
	$-\1_{A}\otimes V_{B}^T \leq m J_{\widetilde{\cN}}^{T_B} \leq \1_{A}\otimes V_{B}^T$. 
\end{corollary}
\begin{proof}
Denote the Choi-Jamio\l{}kowski matrix of $\id_m$ as $J_{m} = \sum_{i,j=0}^{m-1}\ket{ii}\bra{jj}_{A_oB_i}$. Then $J_m^T = J_m$ and the SDP for the minimum error $\omega_{\rm{NS}}(\id_m,\cN)$ can be restated as  
\begin{subequations}\label{re SDP S}
	\begin{align}
	\inf & \quad \ \gamma\\
	\text{\rm s.t.} &\quad  \tr_{B_o} Y_{A_iB_o} \leq \gamma \1_{A_i},\\
	& \quad \ Y_{A_iB_o}\ge  J_{\widetilde{\cN}}-J_\cN,\ Y_{A_iB_o} \geq 0, \\
	& \quad \ J_{\widetilde{\cN}} = \tr_{A_oB_i} (J_{m}\ox\1_{A_iB_o})J_{\Pi}, \label{re condi 1}\\
	& \quad  \ J_{\Pi}\ge 0,\ \tr_{A_oB_o} J_{\Pi} = \1_{A_iB_i}, \hspace{1.55cm} \mathrm{(CP,TP)}\label{re CPTP}\\
	&\quad \, \tr_{A_o}J_{\Pi}=\1_{A_i}/|A_i|\ox\tr_{A_oA_i}J_{\Pi},\label{re A not to B} \hspace{0.8cm} (\mathrm{A} \not \to \mathrm{B})\\
	& \quad \, \tr_{B_o}J_{\Pi}=\1_{B_i}/|B_i|\ox\tr_{B_iB_o}J_{\Pi},\hspace{0.8cm} (\mathrm{B} \not \to \mathrm{A})\label{re B not to A}		
	\end{align}
\end{subequations}
where the dimensions of $A_o$ and $B_i$ are equal to $m$.

The main idea to do the simplification is to utilize the symmetry of the noiseless quantum channel, i.e., the invariance of $J_m$ under any local unitary $U_{A_o}\ox \overline{U}_{B_i}$. Suppose $J_{\Pi}$ is optimal for SDP~\eqref{re SDP S}, we can check that $(U_{A_o} \ox \overline{U}_{B_i}) J_{\Pi} (U_{A_o} \ox \overline{U}_{B_i})^\dagger$ is also optimal. Any convex combination of optimal solutions remains optimal. Thus, without loss of generality we can take~\cite{Rains2001,Leung2015c},
	\begin{align}\label{special form Z}
	J_{\Pi}  = & \int dU (U_{A_o}\ox \overline{U}_{B_i})J_{\Pi}(U_{A_o}\ox \overline{U}_{B_i})^\dagger =\frac{J_m}{m} \ox H_{A_i B_o} + (\1_{A_oB_i} - \frac{J_m}{m}) \ox K_{A_iB_o},
	\end{align}
	where the integral is taken over the Haar measure and $H$, $K$ are operators on system $A_iB_o$. 

	Note that $J_m \cdot J_m = mJ_m$. Taking Eq.~\eqref{special form Z} into the condition~\eqref{re condi 1}, we obtain that
	\begin{align}
		\text{Condition}~\eqref{re condi 1} \iff J_{\widetilde \cN} = m H.
	\end{align}
	Taking Eq.~\eqref{special form Z} into the condition~\eqref{re CPTP}, we have the equivalence
	\begin{align}
	 	\text{Condition}~\eqref{re CPTP} \iff H\geq 0,\ K \geq 0,\ \text{and}\ \tr_{B_o} (H +(m^2-1) K) = m \1_{A_i}.
	 \end{align} 
	Since $J_{\widetilde \cN}$ is the Choi-Jamio\l{}kowski matrix of the effective channel, we have $\tr_{B_o} J_{\widetilde \cN} = \tr_{B_o} m H = \1_{A_i}$. Combining $\tr_{B_o} (H +(m^2-1) K) = m \1_{A_i}$, we have $\tr_{B_o} m K = \1_{A_i}$. This implies that the condition~\eqref{re B not to A} is trivial and 
	\begin{align}		 
		\text{Condition}~\eqref{re A not to B} & \iff H + (m^2-1)K= \1_{A_i}/|A_i| \ox \tr_{A_i}(H+ (m^2-1)K).
	\end{align} 
	So far we have the simplified SDP as
	\begin{subequations}
	\begin{align}
	\inf & \quad \ \gamma\\
	\text{\rm s.t.} &\quad  \tr_{B_o} Y_{A_iB_o} \leq \gamma \1_{A_i},\\
	& \quad \ Y_{A_iB_o}\ge  J_{\widetilde{\cN}}-J_\cN,\ Y_{A_iB_o} \geq 0, \\
	& \quad \ J_{\widetilde{\cN}} = m H_{A_iB_o},\label{cora tmp1}\\
	&\quad\  \tr_{B_o}  J_{\widetilde{\cN}} = \1_{A_i} \\
	& \quad  \  H\geq 0,\ K \geq 0,\ \tr_{B_o} m K = \1_{A_i},\\
	& \quad \ H + (m^2-1)K= \1_{A_i}/|A_i| \ox \tr_{A_i}(H+ (m^2-1)K) \label{cora tmp2}
	\end{align}
\end{subequations}
	Denote $V_{B_o} = m\tr_{A_i}(H + (m^2-1)K)/|A_i|$. By conditions~\eqref{cora tmp2}, we have 
	\begin{align}\label{sdp condition k}
		(m^2-1)mK = \1_{A_i} \ox V_{B_o} - mH.
	\end{align} Together with the condition~\eqref{cora tmp1}, we can eliminate the variables $K$ and $H$ in the above SDP. Finally, replacing the subscript $A_i$ to $A$ and $B_o$ to $B$, we have the desired SDP~\eqref{id error SDP}.

	As for the PPT condition $J_{\Pi}^{T_{B_iB_o}} \geq 0$, we have ${J_m^{T_{B_i}}}/{m}\ox H^{T_{B_o}} + (\1 - {J_m^{T_{B_i}}}/{m}) \ox K^{T_{B_o}} \geq 0$ from Eq.~\eqref{special form Z}. Note that $J_m^{T_{B_i}}$ is the swap operator and we can decompose it into the sum of two orthogonal positive parts, i.e., $J_m^{T_{B_i}} = J_+ - J_-$ where $J_+ \geq 0, J_- \geq 0$ and $J_+ + J_- = \1$. Then the PPT condition is equivalent to $J_+ \ox [H^{T_{B_o}} + (m-1) K^{T_{B_o}}] + J_- \ox [-H^{T_{B_o}} + (m+1) K^{T_{B_o}}] \geq 0$. Thus $-(m-1)K^{T_{B_o}} \leq H^{T_{B_o}} \leq (m+1)K^{T_{B_o}}$. Combining with Eqs.~\eqref{cora tmp1} and~\eqref{sdp condition k}, we have $-\1_{A_i}\otimes V_{B_o}^T \leq m J_{\widetilde{\cN}}^{T_{B_o}} \leq \1_{A_i}\otimes V_{B_o}^T$.

\end{proof}

From the definition of one-shot quantum simulation cost and Corollary~\ref{min error simulation}, we have the following result.
\begin{proposition}\label{one shot simulation cost SDP}
	For any quantum channel $\cN$ and error tolerance $\ve\geq 0$, the one-shot $\ve$-error quantum simulation cost under NS-assisted codes is given by the following SDP,
	\begin{subequations}
		\label{SNS SDP}
		\begin{align}
		S^{(1)}_{\rm NS,\varepsilon}(\cN) = \log \ \inf &\quad \Big\lceil\sqrt{\tr V_{B}}\Big\rceil\\
		\text{\rm s.t.} &\quad \tr_{B} Y_{AB} \leq \ve \1_{A},\\
		&\ \quad Y_{AB} \geq J_{\widetilde{\cN}} - J_{\cN},\ Y_{AB} \ge 0,\\
		&\ \quad J_{\widetilde{\cN}} \geq 0, \ \tr_{B} J_{\widetilde{\cN}} = \1_{A}, \label{J widehat channel}\\
		&\ \quad J_{\widetilde{\cN}} \leq \1_{A}\otimes V_{B}.
		\end{align}
	\end{subequations}
\end{proposition}
It is easy to check that $\delta=\log \lceil x \rceil - \log x \in [0,1]$ for any $x \geq 1$. Thus we can use the least constant $\delta \in [0,1]$ to adjust the r.h.s. of Eq.~\eqref{SNS SDP} to be the logarithm of an integer. That is,
	\begin{subequations}	
		\label{SNS SDP1}
		\begin{align}
		S^{(1)}_{\rm NS,\varepsilon}(\cN) = \delta \ + \ \frac12 \log \ \inf &\quad \tr V_{B}\\
		\text{\rm s.t.} &\quad \tr_{B} Y_{AB} \leq \ve \1_{A},\\
		&\ \quad Y_{AB} \geq J_{\widetilde{\cN}} - J_{\cN},\ Y_{AB} \ge 0,\\
		&\ \quad J_{\widetilde{\cN}} \geq 0, \ \tr_{B} J_{\widetilde{\cN}} = \1_{A}, \label{J widehat channel}\\
		&\ \quad J_{\widetilde{\cN}} \leq \1_{A}\otimes V_{B}.
		\end{align}
	\end{subequations}
Note that the one-shot quantum simulation cost under NS$\cap$PPT-assisted codes is not an SDP, since the objective function $m$ appears in the conditions $\tr V_{B} = m^2$ and $-\1_{A}\otimes V_{B}^T \leq m J_{\widetilde{\cN}}^{T_B} \leq \1_{A}\otimes V_{B}^T$ with different powers. We do not see a way to obtain a linear objective function.

It is also worth mentioning that the zero-error quantum simulation cost was studied by Duan and Winter in~\cite{Duan2016}. The authors show that the zero-error NS-assisted simulation cost is given by the conditional min-entropy of the channel's Choi-Jamio\l{}kowski matrix~\cite[Theorem 2]{Duan2016}. The result we obtained in Proposition~\ref{one shot simulation cost SDP} is more general and can reduce to the zero-error case by letting $\ve = 0$. More specifically, taking $\ve = 0$ will lead to $Y_{AB} = 0$ and thus $J_{\widetilde \cN} = J_{\cN}$. Then we have

\begin{align}	\label{zero error conditional min entropy}
S_{\rm{NS},0}^{(1)}(\cN)
=\frac{1}{2} \log\ \inf \left\{\,\tr V_{B} | \ J_{\cN}\le \1_{A}\ox V_{B}\right\} + \delta,
\end{align}	
where $\delta \in [0,1]$ is the least constant such that the r.h.s. is the logarithm of an integer.
Since the conditional min-entropy is additive (see~\cite{Tomamichel2012}), we have 
\begin{align}
	S_{\rm{NS},0}(\cN) :=& \lim_{n\to \infty} \frac1n S_{\rm NS,0}^{(1)}(\cN^{\ox n})=\frac{1}{2} \log\ \inf \left\{\,\tr V_{B} | \ J_{\cN}\le \1_{A}\ox V_{B}\right\}.
\end{align}

\begin{remark}
	In the next section, we will see that the zero-error NS-assisted simulation cost could also be considered to be the max-information of the channel's \Choi state based on Eqs.~\eqref{channel max information} and~\eqref{cost and max information zero error}. 
\end{remark}


\section{The channel's max-information and its asymptotic equipartition property}\label{max information part}

In this section, we introduce a novel entropy called the channel's smooth max-information and show that it has an operational interpretation regarding the quantum simulation cost of a channel. Furthermore, we prove the asymptotic equipartition property (AEP) of the channel's smooth max-information and explore its close relation to the well-known quantum reverse Shannon theorem (QRST).

Some basic notations will be used in this section. The set of sub-normalized quantum states is denoted as $S_\leq(A):=\{\, \rho \geq 0 \ |\ \tr \rho \leq 1\,\}$. The set of quantum states is denoted as $S_=(A):=\{\,\rho \geq 0\ |\ \tr \rho = 1\,\}$. We denote $\rho_A$ as the reduced state of $\rho_{AB}$, i.e. $\rho_A := \tr_B \rho_{AB}$. The purified distance based on the generalized fidelity is given by~\cite{Tomamichel2012}
\begin{align}
P(\rho,\sigma) := \sqrt{1- F^2(\rho,\sigma)} \ \ \ \text{with} \ \ \ F(\rho,\sigma) := \|\rho^{1/2} \sigma^{1/2}\|_1 + \sqrt{(1-\tr \rho)(1-\tr \sigma)}.
\end{align}
We say $\rho$ and $\sigma$ are $\ve$-close and write $\rho \approx^{\ve} \sigma$ if and only if $P(\rho,\sigma) \leq \ve$.

\subsection{The channel's max-information}

The max-relative entropy of $\rho \in \cS_\leq(A)$ with respect to $\sigma \geq 0$ is defined as~\cite{Datta2009,Renner2005}
\begin{align}\label{max relative entropy}
D_{\max}(\rho\|\sigma) :=  \inf \big\{\,t \ |\ \rho \leq 2^t \cdot \sigma\,\big\},
\end{align}
which is a one-shot generalization of the equantum relative entropy
\begin{align}
	D(\rho\|\sigma):= \tr \rho(\log \rho -\log \sigma)
\end{align}
if $\supp(\rho) \subseteq \supp(\sigma)$ and $+\infty$ otherwise.
The max-information that $B$ has about $A$ for $\rho_{AB}\in \cS_\leq({AB})$ is defined as~\cite{Berta2011a}
\begin{align}\label{max information}
I_{\max}(A:B)_\rho := \inf_{\sigma_B \in \cS_=(B)}D_{\max}(\rho_{AB}\| \rho_A \ox \sigma_B),
\end{align}
which is a one-shot generalization of the quantum mutual information
\begin{align}
I(A:B)_\rho:= \inf_{\sigma_B\in \cS_=(B)} D(\rho_{AB}\|\rho_A \ox \sigma_B).
\end{align}

\begin{definition}
For any quantum channel $\cN_{A'\to B}$ we define the channel's max-information of $\cN$ as 
\begin{align}\label{channel max information}
I_{\max}(A:B)_\cN := I_{\max}(A:B)_{\cN_{A'\to B}(\Phi_{AA'})},
\end{align}
where $\Phi_{AA'} =\frac{1}{|A|}\sum_{i,j=0}^{|A|-1} \ket{i_Ai_{A'}}\bra{j_Aj_{A'}}$ is the maximally entangled state on $AA'$.	
\end{definition} 
\begin{remark}\label{rem: Imax input}
The following argument shows that this definition does not depend on the input state $\Phi_{AA'}$. That is, for any full rank state $\phi_{A'}$ with a purification $\phi_{AA'}=|A|\sqrt{\phi_A} \Phi_{AA'} \sqrt{\phi_A}$, we have
\begin{align}\label{Imax equality}
	I_{\max}(A:B)_\cN = I_{\max}(A:B)_{\cN_{A'\to B}(\phi_{AA'})}.
\end{align}
From the definitions~\eqref{max relative entropy},~\eqref{max information} and~\eqref{channel max information}, we have 
\begin{align}\label{I max definition eq 2}
I_{\max}(A:B)_{\cN} =  \inf \left\{\ t \ \Big| \
 \ \cN_{A'\to B}(\Phi_{AA'}) \leq 2^t \cdot \frac{\1_A}{|A|} \ox \sigma_B,\ \sigma_B \in \cS_=(B)\right\}.
\end{align}
Since 
\begin{align}
\cN_{A'\to B}(\phi_{AA'}) =|A| \cdot \cN_{A'\to B}(\sqrt{\phi_A}\,\Phi_{AA'}\sqrt{\phi_A}) = |A| \cdot\sqrt{\phi_A}\,\cN_{A'\to B}(\Phi_{AA'})\sqrt{\phi_A},
\end{align}
the first condition in~\eqref{I max definition eq 2} is equivalent to $\cN_{A'\to B}(\phi_{AA'}) \leq 2^t \cdot \phi_A \ox \sigma_B$, which implies Eq.~\eqref{Imax equality}.
\end{remark}

\begin{remark}
	Note that the sandwiched R\'{e}nyi version of the channel's mutual information was previously defined in~\cite{Gupta2013,Cooney2014}. Our definition of the channel's max-information is compatible with the more general one and can be recovered in the limit as $\a \to \infty$~\cite{Gour2018}. 
\end{remark}

\vspace{0.2cm}
Comparing Eqs.~\eqref{zero error conditional min entropy} and~\eqref{I max definition eq 2}, we can write the one-shot zero-error quantum simulation cost as the channel's max-information:
\begin{align}\label{cost and max information zero error}
S_{\rm{NS},0}^{(1)}(\cN)=\frac12 I_{\max}(A:B)_\cN + \delta,
\end{align}
where $\delta \in [0,1]$ is the least constant such that the r.h.s. is the logarithm of an integer.
In the following, we show this relation beyond the zero-error case. For this, we define the smoothed version of the channel's max-information.
\begin{definition}
For any quantum channel $\cN$, we define the channel's smooth max-information as 
\begin{align}
I_{\max}^\ve(A:B)_\cN := \inf_{\substack{\frac12\|\widetilde \cN - \cN\|_\di \leq \ve\\ \widetilde \cN \in\ \text{\rm CPTP}(A':B)}}  \quad I_{\max}(A:B)_{\widetilde \cN},
\end{align}	
where $\text{\rm CPTP}(A':B)$ denotes the set of all CPTP maps from $A'$ to $B$. 	
\end{definition}
We show that the one-shot $\ve$-error quantum simulation cost is completely characterized by the channel's smooth max-information. This provides the operational meaning of this new entropy.

\begin{theorem}\label{simulation cost I max}
	For any quantum channel $\cN$ and given error tolerance $\ve \geq 0$, it holds that
	\begin{align}\label{one shot cost I max}
	S^{(1)}_{\rm{NS},\varepsilon}(\cN) = \frac12I_{\max}^\ve(A:B)_\cN + \delta,
	\end{align}
	where $\delta \in [0,1]$ is the least constant such that the r.h.s. is the logarithm of an integer.
\end{theorem}

\begin{proof}
	We first notice that the constraints in Eq.~\eqref{J widehat channel} $J_{\widetilde{\cN}} \geq 0, \ \tr_{B} J_{\widetilde{\cN}} = \1_{A}$  uniquely define a CPTP map $\widetilde \cN$ due to the Choi-Jamio\l{}kowski isomorphism. Applying the SDP of diamond norm in~\eqref{SDP diamond}, we find 
	\begin{subequations}
		\label{one shot SC NS refine}
		\begin{align}
		S^{(1)}_{\rm{NS},\varepsilon}(\cN) = \delta\ + \  \frac{1}{2}\log \ \inf & \quad \tr V_B \\
		\text{\rm s.t.} &\quad \frac{_1}{^2}\|\widetilde \cN - \cN\|_\di \leq \ve,\\
		& \quad \widetilde \cN \in \text{CPTP}(A': B),\\
		&\quad J_{\widetilde \cN} \leq  \1_{A}\otimes V_B.
		\end{align}
	\end{subequations}
	From Eqs.~\eqref{I max definition eq 2}, we know that
\begin{align}\label{max information SDP}
	I_{\max}(A:B)_\cN = \log \inf \big\{\,\tr V_{B}\ |\ J_{\cN}\le \1_{A}\ox V_{B}\,\big\}.
	\end{align}
Combining Eqs.~\eqref{one shot SC NS refine} and \eqref{max information SDP}, we obtain the desired result.
\end{proof}

\begin{remark}
	Note that an alternative way to write Eq.~\eqref{one shot cost I max} is $S^{(1)}_{\rm{NS},\varepsilon}(\cN) = \log \Big\lceil \sqrt{2^{I_{\max}^\ve(A:B)_\cN}}\Big\rceil$.
\end{remark}

\vspace{0.1cm}
\begin{remark}\label{rem: Imax data-processing}
	Since $I_{\max}^\ve(A:B)_\cN$ is introduced as an entropy of the channel $\cN$, it is natural to consider its data-precessing inequality, i.e., $I_{\max}^\ve(A:B)_{\Theta(\cN)} \leq I_{\max}^\ve(A:B)_{\cN}$ for any quantum superchannel $\Theta$~\cite{Chiribella2008}. By the operational meaning above, this is equivalent to $S^{(1)}_{\rm{NS},\varepsilon}(\Theta(\cN)) \leq S^{(1)}_{\rm{NS},\varepsilon}(\cN)$ which can be understood as we need less resources to simulate a quantum
	channel with higher noise. Specifically, this relation can be checked by the definition of quantum channel simulation cost. Suppose the simulation cost of $\cN$ is given by $\log m$. This implies that there exists an NS-assisted code $\Pi$ such that $\frac12 \|\Pi\circ \id_m - \cN\|_\di \leq \ve$. For any quantum superchannel $\Theta$, it can be written~\cite{Chiribella2008} as $\Theta(\cN) =  \cR_{B_o E\to \bar B} \circ (\cN_{A_i \to B_o}\ox \id_E)\circ \cF_{\bar A \to A_i E}$ with pre-processing channel $\cF_{\bar A \to A_i E}$ and post-processing channel $\cR_{B_o E\to \bar B}$. Note that the diamond norm is multiplicative under tensor product, sub-multiplicative under composition and equals to one for any quantum channels~\cite{Kitaev1997}. Then it holds $\frac12 \|\Theta(\Pi\circ \id_m) - \Theta(\cN)\|_\di \leq \frac12 \|\Pi\circ \id_m - \cN\|_\di \leq \ve$. This shows that we can use the noiseless channel $\id_m$ to simulate $\Theta(\cN)$ under the NS-assisted code $\Theta\circ \Pi$ within $\ve$ error. By definition, we have $S^{(1)}_{\rm{NS},\varepsilon}(\Theta(\cN)) \leq \log m = S^{(1)}_{\rm{NS},\varepsilon}(\cN)$.

\end{remark}

\begin{remark}
	Note that the \Choi matrix of a constant channel $\cM(\rho)=\sigma$, $\forall \rho$ is given by $J_{\cM} = \1_A \ox \sigma_B$. Thus from the perspective of quantum resource theory, the channel's smooth max-information can be written as the ``distance'' between the given channel $\cN$ and the set of constant channels 
	\begin{align}
   \boldsymbol\cG := \big\{\,\cM \in \text{CPTP}(A:B)\ \big|\ \exists \ \sigma \ \text{\rm s.t.}\  \cM(\rho) = \sigma, \forall\, \rho \,\big\}.
\end{align}
More specifically, for any quantum channel $\cN$, we have 
\begin{align}\label{distance characterization}
 	I_{\max}^\ve(A:B)_\cN = \min_{\cM \in \boldsymbol \cG} D^\ve_{\max}(\cN\|\cM),
 \end{align}
 where 
 \begin{align}
 	D^\ve_{\max}(\cN\|\cM):= \inf_{\substack{\frac12\|\widetilde \cN - \cN\|_\di \leq \ve\\ \widetilde \cN \in \text{CPTP}(A':B)}} D_{\max}\big(\widetilde \cN\|\cM\big)
 \end{align} 
 and $D_{\max}(\cN\|\cM):= D_{\max}(J_{\cN}\|J_{\cM})$ with \Choi matrices $J_{\cN}$, $J_{\cM}$. Since the max-relative entropy is closely related with the robustness~\cite{Datta2009b}---the minimal mixing required to make the given resource useless, we can also define the channel's analogy of robustness as (see also \cite{Diaz2018})
 \begin{align}
   \cR_g(\cN):=\inf\left\{t \geq 0 \ \bigg|\ \exists \ \cM \in \text{CPTP}(A:B)\ \text{\rm s.t.}\ \frac{\cN+t\cM}{1+t} \in \boldsymbol\cG \right\},
 \end{align}
 and its smoothed version 
 \begin{align}
 	R^\ve_g(\cN):= \inf_{\substack{\frac12\|\widetilde \cN - \cN\|_\di \leq \ve\\ \widetilde \cN \in \text{CPTP}(A':B)}} R_g\big(\widetilde \cN\big).
 \end{align}
 Then the channel's smooth max-information can be written as
 \begin{align}
 	I_{\max}^\ve(A:B)_{\cN} = \log (1+\cR_g^\ve(\cN)).
 \end{align}
\end{remark}

\subsection{Asymptotic equipartition property and the quantum reverse Shannon theorem}

In the framework of quantum channel simulation, the quantum  capacity is given by the optimal rate of using $\cN$ to simulate the qubit noiseless channel $\id_2$, while the channel simulation cost is given by the optimal rate of using $\id_2$ to simulate the channel $\cN$. Thus, it operationally holds  that
\begin{align} \label{equality chain}
Q_{\rm E}(\cN) \leq Q_{\rm{NS}}(\cN) \leq S_{\rm{NS}}(\cN) \leq S_{\rm E}(\cN),
\end{align}
where the above four notations represent entanglement-assisted quantum capacity, NS-assisted quantum capacity, NS-assisted quantum simulation cost and entanglement-assisted quantum simulation cost, respectively. The quantum reverse Shannon theorem~\cite{Bennett2014,Berta2011a} states that the quantum simulation cost is equal to its quantum capacity under entanglement-assistance, i.e., $Q_{\rm E}(\cN) = S_{\rm E}(\cN)$. In the following, the quantum reverse Shannon theorem under NS-assistance means that $Q_{\rm{NS}}(\cN) = S_{\rm{NS}}(\cN)$. 

The AEP of the channel's smooth max-information is the claim that
\begin{align}\label{AEP equation}
\lim_{\ve \to 0}\lim_{n\to \infty} \frac{1}{n}  I_{\max}^\ve(A:B)_{\cN^{\ox n}} =  I(A:B)_{\cN},
\end{align}
where 
\begin{align}
I(A:B)_{\cN} := \max_{\rho_A\in S_=(A)} I(A:B)_{\cN_{A'\to B}(\phi_{AA'})}
\end{align} is the mutual information of the quantum channel and $\phi_{AA'}$ is a purification of the state $\rho_A$.
Based on the operational interpretation in Theorem~\ref{simulation cost I max}, the definition of asymptotic simulation cost in Eq.~\eqref{eq: definition asymp cost} and the known result that $Q_{\rm E}(\cN) = \frac12I(A:B)_{\cN}$~\cite{Bennett2002}, we have 
\begin{align}
\rm{AEP}~\eqref{AEP equation} \iff Q_E(\cN) = S_{\rm NS}(\cN).
\end{align}
Hence the QRST implies the AEP for the channel's smooth max-information. On the other hand, we can directly prove the AEP, as is done in Theorem~\ref{AEP directly}. This proof then implies the QRST in the presence of NS correlations.

In the following, we utilize the smooth max-information of a quantum state and its variation:
\begin{align}
I_{\max}^\ve(A:B)_\rho := &\inf_{\widetilde \rho \approx^\ve \rho} I_{\max}(A:B)_{\widetilde \rho},\label{eq:IMAX}\\
\widehat I_{\max}^\ve(A:B)_\rho := &\inf_{\substack{\widetilde \rho \approx^\ve \rho\\ \widetilde \rho_A = \rho_A}} I_{\max}(A:B)_{\widetilde \rho}.\label{eq:I max variation}
\end{align}
The first smooth max-information is most often used in the literature~\cite{Berta2011a,Ciganovic2014}. The second variation naturally appears in our discussion of the channel simulation problem. The restricted smoothing such that the marginal state is fixed comes from the definition of diamond norm where the reference system of the input state is untouched. Using ideas from~\cite{Anshu2018a}, the following lemma shows that these two quantities are equivalent up to some correction terms. The proof can be found in Appendix~\ref{tech lemmas}.

\begin{lemma}
For any quantum state $\rho_{AB}$ and $\ve \in (0,1)$, it holds
\begin{align}\label{eq:I max lemma}
\widehat I_{\max}^{\ve}(A:B)_\rho \leq I_{\max}^{\ve/6}(A:B)_\rho + g(\ve)\quad\text{with}\quad g(\ve) = \log(1+72/\ve^2).
\end{align}
\end{lemma}

\begin{theorem}\label{AEP directly}
For any quantum channel $\cN$ we have the AEP for the channel's smooth max-information:
\begin{align}
\lim_{\ve \to 0}\lim_{n\to \infty} \frac{1}{n}  I_{\max}^\ve(A:B)_{\cN^{\ox n}} =  I(A:B)_{\cN}.
\end{align}
\end{theorem}

\begin{proof}
	The proof strategy is as follows. We first use the post-selection technique to show that the channel's smooth max-information is upper bounded by the quantity in Eq.~\eqref{eq:I max variation}. By Eq.~\eqref{eq:I max lemma} we can then use the basic properties of the smooth max-information developed in~\cite{Berta2011a} to show one direction of the proof. The other direction can be proved via the continuity of the mutual information of quantum states. 
	
	Consider $n$ uses of the channel $\cN$ and let $\o_{RAA'}^n$ be the purification of the de Finetti state $\o^n_{AA'} = \int \sigma_{AA'}^{\ox n} d(\sigma_{AA'})$ with pure states $\sigma_{AA'} = \ket{\sigma}\bra{\sigma}_{AA'}$ and $d(\cdot)$ the measure on the normalized pure states induced by the Haar measure. Furthermore we can assume without loss of generality that $|R| \leq (n+1)^{|A|^2 -1}$~\cite{Christandl2009}. Note that $\w^n_{A'}$ is a full rank state. We have the following chain of inequalities
	\begin{adjustwidth}{-1em}{}
	\begin{subequations}
		\begin{align}
		& I_{\max}^\ve(A:B)_{\cN^{\ox n}} \notag\\
		= & \inf \left\{I_{\max}(RA:B)_{\widetilde \cN^{n}(\o_{RAA'}^n)} \ \Big|\ \frac12 \left\|\widetilde \cN^n - \cN^{\ox n}\right\|_\di \leq \ve,\, \widetilde \cN^n \in \text{CPTP}(A'^n:B^n)\right\}\label{a}\\
        \leq & \inf \left\{I_{\max}(RA:B)_{\widetilde \cN^{n}(\o_{RAA'}^n)} \ \Big|\ \frac12 \left\|\widetilde \cN^n - \cN^{\ox n}\right\|_\di \leq \ve,\, \widetilde \cN^n \in \text{Perm}(A'^n:B^n)\right\}\label{b}\\
        \leq & \inf \left\{I_{\max}(RA:B)_{\widetilde \cN^{n}(\o_{RAA'}^n)} \ \Big|\ \frac12 \left\|(\widetilde \cN^n - \cN^{\ox n})(\o_{RAA'}^n)\right\|_1 \leq \ve_1,\, \widetilde \cN^n \in \text{Perm}(A'^n:B^n)\right\}\label{c}\\
		\leq & \inf \left\{I_{\max}(RA:B)_{\widetilde \cN^{n}(\o_{RAA'}^n)} \ \Big|\ \widetilde \cN^n (\o_{RAA'}^n) \approx^{\ve_2} \cN^{\ox n}(\o_{RAA'}^n),\, \widetilde \cN^n \in \text{Perm}(A'^n:B^n)\right\}\label{d},
		\end{align}
	\end{subequations}
	\end{adjustwidth}
	where $\ve_1 = \ve(n+1)^{-(|A'|^2-1)}$ and we take $\ve_2 = \ve_1$. Note that the discussion works for any $0 < \ve_2 \leq \ve_1$. 
    In~\eqref{a}, we choose $\o^n_{RAA'}$ as the input state of the channel's max-information \eqref{Imax equality}.
    In~\eqref{b}, we restrict the channel $\widetilde \cN^n$ to be permutation invariant, where $\text{Perm}(A'^n:B^n):=\{\cN^n\in \text{CPTP}(A'^n:B^n) \,|\, \pi^n_B \circ \cN^n \circ \pi^n_{A'} = \cN,\ \text{for all permutation}\ \pi^n\}$ denotes the set of all permutation invariant channels.
	In~\eqref{c}, we use post-selection technique~(see~\cite[Prop.~D.4]{Berta2011a}), which relaxes the diamond norm to the trace norm. In~\eqref{d}, we replace the trace norm with the purified distance due to the inequality $\frac12\|\rho-\sigma\|_1\leq P(\rho,\sigma)$. 
	
	Exploiting the permutation invariance of $\cN^{\ox n}$, we know that the optimal solution of~\eqref{d} can be still taken at a permutation invariant channel even if we relax the set $\text{Perm}(A'^n:B^n)$ to all CPTP maps, which is then equivalent to optimize over all quantum states with marginal $\o^n_{RA}$ and $\ve_2$-close to $\cN^{\ox n}(\o_{RAA'}^n)$. Specifically, from Lemma~\ref{two sets same 1} in Appendix~\ref{tech lemmas}, we know that the optimization in~\eqref{d} is equivalent to 
	\begin{align}\label{two sets}
	\inf_{\sigma^n_{RAB} \in K} I_{\max}(RA:B)_{\sigma^n_{RAB}} \quad \text{with}\quad K:= \left\{\sigma^n_{RAB} \Big|\, \sigma^n_{RAB} \approx^{\ve_2} \cN^{\ox n}(\o_{RAA'}^n), \sigma^n_{RA} = \o^n_{RA} \right\},
	\end{align}
    which is exactly the definition of $\widehat I_{\max}^{\ve_2}(RA:B)_{\cN^{\ox n}(\o_{RAA'}^n)}$.
    Thus we have
	\begin{align}
	I_{\max}^\ve(A:B)_{\cN^{\ox n}} \leq \widehat I_{\max}^{\ve_2}(RA:B)_{\cN^{\ox n}(\o_{RAA'}^n)}.
	\end{align}
	From Eq.~\eqref{eq:I max lemma}, denote $\ve_3 = \ve_2/6$, we have
	\begin{align}\label{channel I max and global I max}
	I_{\max}^\ve(A:B)_{\cN^{\ox n}} \leq I_{\max}^{\ve_3}(RA:B)_{\cN^{\ox n}(\o_{RAA'}^n)} + g(\ve_2).
	\end{align}
	Then we can use some known properties of the smooth max-information from~\cite{Ciganovic2014,Berta2011a}, which leads to
	\begin{subequations}
		\begin{align}
		I_{\max}^{\ve_3}(RA:B)_{\cN^{\ox n}(\o_{RAA'}^n)}
		& \leq I^{\ve_4}_{\max}(B:RA)_{\cN^{\ox n}(\o^n_{RAA'})} + f(\ve_4)\label{g}\\
		& \leq I^{\ve_4}_{\max}(B:A)_{\cN^{\ox n}(\o^n_{AA'})}  + 2\log |R| + f(\ve_4)\label{h}\\
		& = I^{\ve_4}_{\max}(B:A)_{\cN^{\ox n}(\sum_{i\in I} p_i (\sigma_{AA'}^i)^{\ox n})}  + 2\log |R| + f(\ve_4) \label{i}\\
		& \leq \max_{\sigma_{AA'}^i}I^{\ve_4}_{\max}(B:A)_{\cN^{\ox n}((\sigma_{AA'}^i)^{\ox n})} + \log |I|+ 2\log |R| + f(\ve_4) \label{j}\\
		& \leq \max_{\sigma_{AA'}}I^{\ve_4}_{\max}(B:A)_{\cN(\sigma_{AA'})^{\ox n}} + \log |I|+ 2\log |R| + f(\ve_4),\label{k}
		\end{align}
	\end{subequations}
where $\ve_4 = \ve_3/2$, $f(\ve) = \log (\frac{1}{1-\sqrt{1-\ve^2}}+\frac{1}{1-\ve})$ and $|I| = (n+1)^{2|A||A'|-2}$. In the first line, we swap the system order according to \cite[Corollary 5]{Ciganovic2014}. In the second line, we get rid of purification system $R$ according to \cite[Lemma B.12]{Berta2011a}. In the third line, we express the integral $\o^n_{AA'} = \int \sigma_{AA'}^{\ox n} d(\sigma_{AA'})$ into convex combination of finite number of operators according to \cite[Corollary D.6]{Berta2011a}. In the forth line, we use the quasi-convexity of the smooth max-information \cite[Lemma B.21]{Berta2011a}. In the last line, we relax the maximization to all pure states $\sigma_{AA'}$. 

Combining Eqs.~\eqref{channel I max and global I max}, \eqref{k} and the AEP for the smooth max-information from~\cite[Lemma B.24]{Berta2011a}, we get
\begin{align}\label{upper bound}
    \lim_{\ve\to 0} \lim_{n\to \infty} \frac1n I_{\max}^\ve(A:B)_{\cN^{\ox n}} \leq  \max_{\sigma_{AA'}} I(A:B)_{\cN(\sigma_{AA'})} = I(A:B)_{\cN}.
\end{align}

On the other hand, suppose the optimal solution of $I(A:B)_{\cN}$ is taken at $\rho_{A'}$ with a purification~$\phi_{AA'}$. Since we can always find a full rank state that is arbitary close to $\rho_{A'}$, thus it gives the mutual information arbitary close to $I(A:B)_\cN$ due to the continuty. In the following, we can assume that $\rho_{A'}$ is of full rank without loss of generality and have the  chain of inequalities
\begin{subequations}
\begin{align}
    I^\ve_{\max}(A:B)_{\cN^{\ox n}} & = \inf_{\substack{\frac12\|\widetilde \cN^n- \cN^{\ox n}\|_\di \leq \ve\\ \widetilde \cN^n \in\ \text{CPTP}(A'^n:B^n)}} \inf_{\sigma_{B}^n \in \cS_=(B^{\ox n})} D_{\max}(\widetilde \cN^n_{A'\to B}(\phi_{AA'}^{\ox n})\|  \phi_{A}^{\ox n}  \ox \sigma_{B}^n)\\
    & \geq \inf_{\substack{\frac12\|\widetilde \cN^n- \cN^{\ox n}\|_\di \leq \ve\\ \widetilde \cN^n \in\ \text{CPTP}(A'^n:B^n)}} \inf_{\sigma_{B}^n \in \cS_=(B^{\ox n})} D(\widetilde \cN^n_{A'\to B}(\phi_{AA'}^{\ox n})\|  \phi_{A}^{\ox n}  \ox \sigma_{B}^n)\\
    & = \inf_{\substack{\frac12\|\widetilde \cN^n- \cN^{\ox n}\|_\di \leq \ve\\ \widetilde \cN^n \in\ \text{CPTP}(A'^n:B^n)}} I(A:B)_{\widetilde \cN^n_{A'\to B}(\phi_{AA'}^{\ox n})}\\
    & \geq I(A:B)_{\cN^{\ox n}_{A'\to B}(\phi_{AA'}^{\ox n})} - (2 n \ve \log |A| + (1+\ve) h_2(\ve/(1+\ve)))\\
    & = n I(A:B)_{\cN_{A'\to B}(\phi_{AA'})} - (2 n \ve \log |A| + (1+\ve) h_2(\ve/(1+\ve)))\\
    & = n I(A:B)_{\cN} - (2 n \ve \log |A| + (1+\ve) h_2(\ve/(1+\ve)))
\end{align}
\end{subequations}
where $h_2(\cdot)$ is the binary entropy. In the second line, we use the fact that max-relative entropy is never smaller than the relative entropy \cite{Datta2009}. The third line follows from the definition of the mutual information of a quantum state. The fourth line follows from the continuity of quantum mutual information in Lemma~\ref{continuity of mutual information} (Appendix~\ref{tech lemmas}). The fifth line follows from the additivity of quantum mutual information. The last line follows from the assumption that $\phi_{AA'}$ is the optimizer of $I(A:B)_{\cN}$. Finally, we have
\begin{align}\label{lower bound}
\lim_{\ve\to 0} \lim_{n\to \infty} \frac1n I_{\max}^\ve(A:B)_{\cN^{\ox n}} \geq  I(A:B)_{\cN}.
\end{align}
Combining Eqs.~\eqref{upper bound} and~\eqref{lower bound}, we conclude the claim.
\end{proof}

\vspace{0.2cm}
After this work there is an alternative proof of Eq.~\eqref{upper bound} given by Gour and Wilde in~\cite{Gour2018}. Their proof uses the sandwiched R\'{e}nyi mutual information and its relation with the max-information, different from the post-selection technique we use in this work.


\section{Examples} \label{sec:examples}

In this section, we apply our results to some basic channels. For classical channels, the one-shot $\ve$-error quantum simulation cost can be given by a linear program as shown in Eq.~\eqref{CC LP} (Appendix~\ref{example app}). Using the symmetry of the quantum depolarizing channel, we can also simplify its $n$-shot simulation cost as a linear program. Moreover, the zero-error simulation cost of various channels can be solved analytically.

\vspace{0.2cm}
\noindent\textbf{Example 1.}~The quantum depolarizing channel is given by $\cD_p(\rho) = (1-p)\rho+p\cdot \frac{\1}{d}$ with dimension~$d$. Its \Choi matrix $J_{\cD_p}$ commutes with any local unitary $U\ox \overline U$ and $J_{\cD_p}^{\ox n}$ is invariant under any permutation of the tensor factors.
Exploiting these symmetries, we can simplify the SDP (\ref{SNS SDP}) for $\cD_p^{\ox n}$ to a linear program~\eqref{DP LP} in Appendix~\ref{example app}. Numerical implementation is shown in Fig.~\ref{DPchannel LP plot}. We can see that as the number of channel uses $n$ increases, the average quantum simulation cost will approach to its entanglement-assisted quantum capacity~\cite{Bennett1999}, i.e., half of the quantum mutual information of the channel.

\begin{figure}[H]
\centering
\begin{tikzpicture}
\node[inner sep=0pt] at (0,0)
{\includegraphics[width=9cm]{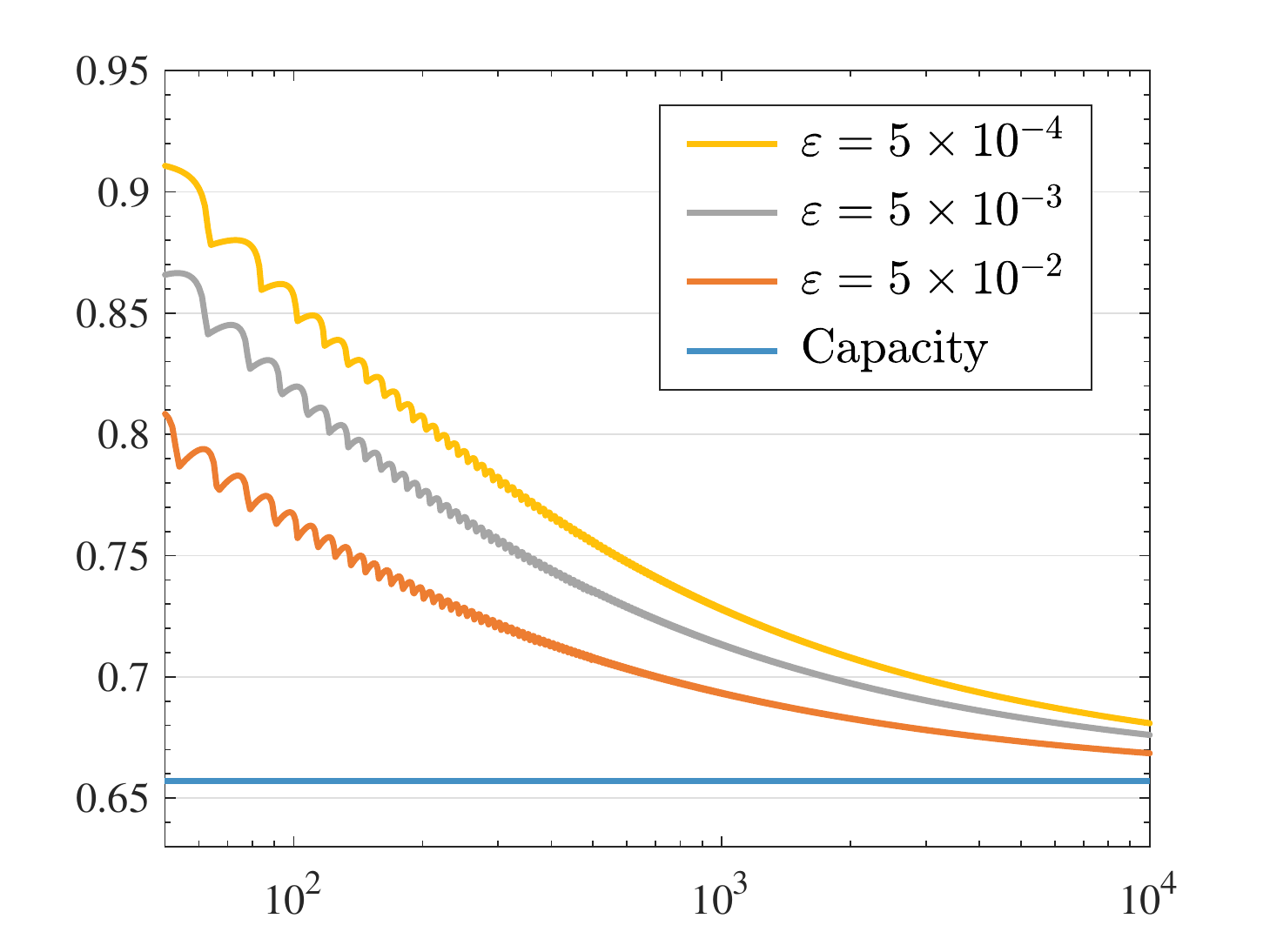}};
\node[] at (0.2,-3.5) {Number of channel uses, $n$};
\node[rotate=90] at (-4.4,0) {Qubits per channel use};
\end{tikzpicture}
\caption{Exact value by the linear program~\eqref{DP LP} of the average simulation cost for three different error tolerances $\ve \in \{5\times10^{-4},5\times10^{-3},5\times10^{-2}\}$ and the qubit depolarizing channel with failure probability $p=0.15$. The lowest line marks the entanglement-assisted quantum capacity of the channel (roughly 0.657 qubits per channel use).}
\label{DPchannel LP plot}
\end{figure}

Recall that the primal and dual SDPs of the zero-error  simulation cost are given by \cite{Duan2016}
\begin{align}
\textbf{Primal:} \quad & S_{\rm{NS},0}(\cN) =\frac{1}{2} \log \inf \big\{\tr V_B\ \big|\ J_{\cN}\le \1_A\ox V_B\,\big\},\\
\textbf{Dual:} \quad & S_{\rm{NS},0}(\cN)=\frac{1}{2} \log \sup \big\{\tr J_\cN X_{AB}\ \big|\ \tr_A X_{AB} \le \1_B, X_{AB} \geq 0\,\big\}.
\end{align}
We study some fundamental channels and show their analytical solutions by explicitly constructing feasible solutions in both primal and dual problems, respectively. Using the weak duality, we can argue that the feasible solutions we construct are optimal.

\vspace{0.2cm}
\noindent\textbf{Example 2.}~The quantum depolarizing channel is $\cD_p(\rho) = (1-p)\rho + p\cdot \frac{\1}{d}$ with dimension $d$. Taking \begin{align}
V_B = (d(1-p)+\frac{p}{d})\1_B, \quad \text{and} \quad   X_{AB} = \sum_{i,j=0}^{d-1} \ket{ii}\bra{jj},
\end{align} 
in the primal and dual problems respectively, we can verify that they are feasible solutions. Thus, we have
\begin{align}
\frac12 \log (d^2(1-p)+p) = \frac{1}{2}\log \tr J_{\cD_p} X_{AB} \leq S_{\rm{NS},0}(\cD_p) \leq \frac12 \log \tr V_B = \frac12 \log (d^2(1-p)+p).
\end{align}
We find that
\begin{align}\label{dp ns zero rate}
S_{\rm{NS},0}(\cD_p) = \frac12 \log (d^2(1-p)+p).
\end{align}

\vspace{0.2cm}
\noindent\textbf{Example 3.}~The amplitude damping channel is $\cN_r(\rho)=\sum_{i=0}^1 E_i \rho E_i^\dag$
with $E_0=\ketbra{0}{0}+\sqrt{1-r}\ketbra{1}{1}$, $E_1=\sqrt{r}\ketbra{0}{1}$
and $0\leq r\leq 1$. The optimal solutions are given by 
\begin{align}
V_B=(1+\sqrt{1-r})\ket 0 \bra 0+(\sqrt{1-r}+1-r)\ket 1 \bra 1 \ \ \text{and} \ \ X_{AB}=(\ket {00} +\ket{11})(\bra {00}+\bra {11}).
\end{align}
We find that
\begin{align}
S_{\rm{NS},0}(\cN_r)=\frac{1}{2}\log (2(1+\sqrt{1-r})-r).
\end{align}

\vspace{0.2cm}
\noindent\textbf{Example 4.}~The dephasing channel is $\cZ_p(\rho) = (1-p)\rho + p Z\rho Z$ with $Z = \ket{0}\bra{0}-\ket{1}\bra{1}$. The optimal solutions are given by 
\begin{align}
V_B = (|2p-1|+1)\1_B\quad\text{and}\quad X_{AB} = (\ket{00}+\ket{11})(\bra{00}+\bra{11}).
\end{align} 
We find that
\begin{align}
S_{\rm{NS},0}(\cZ_p) = \frac{1}{2}\log (|4p-2|+2).
\end{align}

\vspace{0.2cm}
\noindent\textbf{Example 5.}~The quantum erasure channel is $\cE_p(\rho) = (1-p)\rho + p\ket{e}\bra{e}$ with $\ket{e}$ orthogonal to the input Hilbert space. The optimal solutions are given by 
\begin{align}
V_B = d(1-p)\sum_{i,j=0}^{d-1}\ket{i}\bra{i} + p \ket{d}\bra{d}\quad\text{and}\quad X_{AB} = \sum_{i,j=0}^{d-1}\ket{ii}\bra{jj}+\frac1d\sum_{i=0}^{d-1}\ket{i}\bra{i}\ox \ket{d}\bra{d}.
\end{align} 
We find that
\begin{align}\label{erasure channel ns zero rate}
S_{\rm{NS},0}(\cE_p) = \frac{1}{2}\log (d^2(1-p)+p).
\end{align}
 
 Finally, we plot the zero-error NS-assisted channel simulation cost of these four channels as a function of their noise parameters in the following Figure~\ref{example channel plot}. The figure is plotted for the qubit case, i.e, $d =2$. Note that the quantum depolarizing channel and the quantum erasure channel have exactly the same rate given by Eqs.~\eqref{dp ns zero rate} and~\eqref{erasure channel ns zero rate}.

\begin{figure}[t]
\centering
\begin{tikzpicture}
\node[inner sep=0pt] at (0,0)
{\includegraphics[width=9cm]{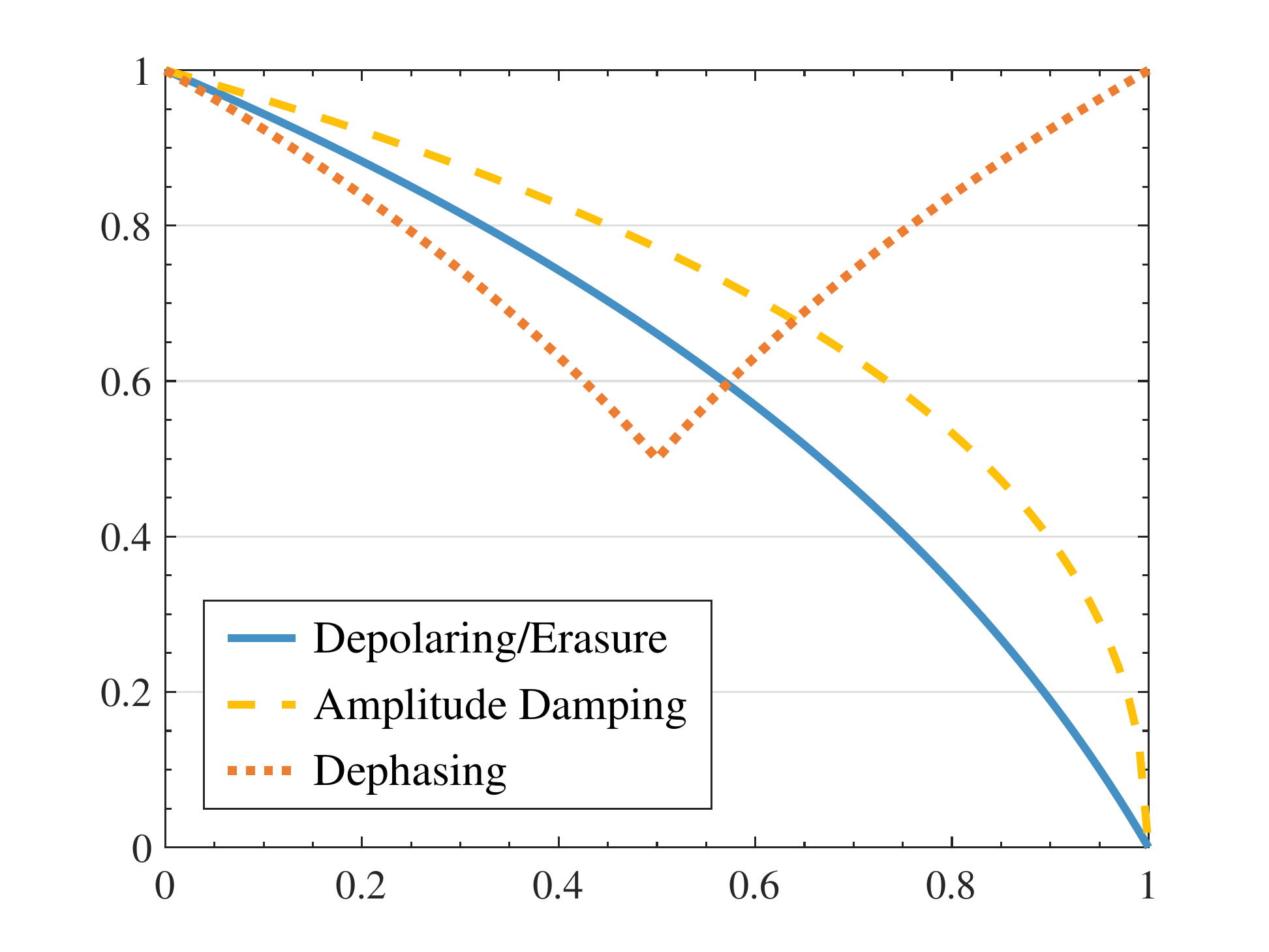}};
\node[] at (0.2,-3.3) {Channel parameter};
\node[rotate=90] at (-4.1,0) {Channel simulation cost (qubits)};
\end{tikzpicture}
\caption{The zero-error NS-assisted channel simulation cost  of the qubit depolarizing channel / qubit erasure channel, qubit amplitude damping channel and qubit dephasing channel as a function of each channel's noise parameter. }
\label{example channel plot}
\end{figure}

\section{Discussion}

Since the entanglement-assisted capacity allows a single-letter characterization, it is natural to consider a second-order refinement thereof. A second-order expansion of an achievable rate was established in~\cite{Datta2016} but no matching second-order converse bound is known. Our one-shot NS-assisted quantum simulation cost and the channel's smooth max-information may provide some insights in this direction.

Suppose a quantum channel $\cN$ can be used to simulate a noiseless channel $\id_{m_1}$ with dimension $m_1$ and on the other hand it requires a noiseless channel $\id_{m_2}$ with dimension $m_2$ to simulate itself, then we necessarily have $m_2 \geq m_1$ by the definition of simulation. This means that the simulation cost of a channel operationally provides a converse bound for its channel capacity. However this approach does not provide a tighter bound than the NS-assisted capacity in the one-shot and asymptotic setting (see, e.g., \cite{Leung2015c,Matthews2012,Wang2016g}).

From Eq.~\eqref{distance characterization}, the AEP of the channel's smooth max-information can be equivalently written as 
\begin{align}
	\lim_{\ve \to 0} \lim_{n\to \infty} \frac1n \inf_{\cM^n \in \boldsymbol \cG_n} D^\ve_{\max}(\cN^{\ox n}\|\cM^n) = \inf_{\cM \in \boldsymbol \cG} D(\cN\|\cM),
\end{align}
where $\boldsymbol \cG_n$ represents the set of constant channels from $A'^n$ to $B^n$ and the channel divergence is defined as~\cite{Cooney2014,leditzky2018approaches}
$D(\cN\|\cM):= \max_{\rho_A}D(\cN_{A'\to B}(\phi_{AA'})\|\cM_{A'\to B}(\phi_{AA'}))$,
with $\phi_{AA'}$ a purification of $\rho_A$. An interesting question is to consider the channel's AEP beyond the set of constant channels $\boldsymbol\cG_n$, such as the singleton $\{\cM^{\ox n}\}$~\footnote{This problem has been independently found and presented by Andreas Winter in the open problems session at the Rocky Mountain Summit on Quantum Information. See \url{http://jila.colorado.edu/rmsqi/open_problems/open-problem-winter.jpg}.}. Can we obtain a channel's generalization of quantum Stein's lemma~\cite{Hiai1991}? A partial progress to this problem regarding classical-quantum channels has been recently given in~\cite{berta2018amortized}.


\section*{Acknowledgements}

We are grateful to Runyao Duan for discussions. MB thanks the Centre for Quantum Software and Information at the University of Technology Sydney for their hospitality while part of this work was done. MT acknowledges an Australian Research Council Discovery Early Career Researcher Award, project No. DE160100821.


\bibliographystyle{IEEEtran}
\bibliography{arXiv}


\begin{appendices}

\section{Linear programs} \label{example app}

For any classical channel $\cN(y|x)$, the SDP~\eqref{SNS SDP} will naturally reduce to a linear program. Specifically, its one-shot simulation cost is given by a linear program,
\begin{subequations}\label{CC LP}
\begin{align}
S^{(1)}_{\rm{NS},\ve}(\cN) = \log\ \inf &\quad \Big\lceil \sqrt{\sum V_y} \Big\rceil\\
\text{\rm s.t.} &\quad\ Y_{xy} \geq  \widetilde\cN(y|x) - \cN(y|x), Y_{xy}\ge 0, \forall x,y,\\
&\quad\ \widetilde\cN(y|x) \ge 0, \forall x,y, \ \sum\nolimits_y \widetilde\cN(y|x) =1, \forall x,\\
&\quad\ \widetilde\cN(y|x) \leq  V_y, \forall x, y,\ \sum\nolimits_{y}Y_{xy} \leq \ve, \forall x.
\end{align}
\end{subequations}

For the quantum depolarizing channel $\cD_p(\rho) = (1-p)\rho + p\cdot \frac{\1}{d}$, its \Choi matrix is given by $J_{\cD_p} = q_1 \Phi_d + q_2 \Phi_d^\perp$ where $q_1 = d(1-p)+\frac{p}{d}$, $q_2 = \frac{p}{d}$ and $\Phi_d$ is the maximally entangled state with dimension $d$, $\Phi_d^\perp = \1 - \Phi_d$.
Then we have
\begin{align}
J_{\cD_p}^{\ox n} = \sum_{k=0}^{n} p_k P_k^n(\Phi_d,\Phi_d^\perp) \quad \text{with}\quad p_k = q_1^{k}q_2^{n-k},
\end{align}
and $P_k^n(\Phi_d,\Phi_d^\perp)$ denotes the summation of $n$-fold tensor products of $\Phi_d$ and $\Phi_d^\perp$ with exactly $k$ factors of $\Phi_d$. For example,
$P_1^3(\Phi_d,\Phi_d^\perp) = \Phi_d^\perp \otimes \Phi_d^\perp \otimes \Phi_d + \Phi_d^\perp  \otimes \Phi_d \otimes \Phi_d^\perp + \Phi_d \otimes \Phi_d^\perp \otimes \Phi_d^\perp.$
Due to the symmetries of $J_{\cD_p}^{\ox n}$, we can take the optimal solution in SDP~\eqref{SNS SDP} in form of
\begin{align}
    J_{\widetilde \cN^n} = \sum_{k=0}^{n} r_k P_k^n(\Phi_d,\Phi_d^\perp),\quad  Y = \sum_{k=0}^{n} y_k P_k^n(\Phi_d,\Phi_d^\perp), \quad \text{and} \quad V = s  \1_{d}^{\ox n}.
\end{align}
Then we have the LP as follows,
\begin{subequations}\label{DP LP}
\begin{align}
S^{(1)}_{\rm NS,\varepsilon}(\cD_p^{\ox n}) = \log\ \inf &\quad \Big\lceil\sqrt{d^n \cdot s}\Big\rceil\\
\text{\rm s.t.} &\quad y_k - r_k + p_k \geq 0, y_k \geq 0, 0 \leq r_k \leq s, \forall k\\
&\quad \sum_{k=0}^n \binom{n}{k}(\frac{1}{d})^k(d-\frac{1}{d})^{n-k} r_k = 1,\\
&\quad \sum_{k=0}^n \binom{n}{k}(\frac{1}{d})^k(d-\frac{1}{d})^{n-k} y_k \leq \ve.
\end{align}
\end{subequations}


\section{Technical lemmas}\label{tech lemmas}

\begin{lemma} 
For any quantum state $\rho_{AB}$ and $\ve \in (0,1)$, it holds
\begin{align}\label{appen eq:I max lemma}
\widehat I_{\max}^{\ve}(A:B)_\rho \leq I_{\max}^{\ve/6}(A:B)_\rho + g(\ve)\quad\text{with}\quad g(\ve) = \log(1+72/\ve^2).
\end{align}
\end{lemma}
\begin{proof}
	Recall the definitions of the smooth max-information for quantum states:
	\begin{align}		
		\widehat I_{\max}^\ve(A:B)_\rho & = \inf_{\substack{\widetilde \rho \approx^\ve \rho\\ \widetilde \rho_A = \rho_A}} \inf_{\sigma_B} D_{\max}(\widetilde \rho_{AB}\| \rho_A \ox \sigma_B),\\		
		I_{\max}^\ve(A:B)_\rho & = \inf_{\widetilde \rho \approx^\ve \rho} \inf_{\sigma_B} D_{\max}(\widetilde \rho_{AB}\|\widetilde \rho_A \ox \sigma_B).
	\end{align}
	In~\cite{Anshu2018a} the authors also discuss the variation
	\begin{align}
		\widetilde I_{\max}^\ve(A:B)_\rho & = \inf_{\widetilde \rho \approx^\ve \rho} \inf_{\sigma_B} D_{\max}(\widetilde \rho_{AB}\|\rho_A \ox \sigma_B),
	\end{align}
	where the marginal state in the second term is fixed to be $\rho_A$.
	It was shown in~\cite{Anshu2018a} that for any quantum state $\rho_{AB}$ and $\ve \in (0,1)$ we have 
	\begin{align}
		\widehat I_{\max}^\ve(A:B)_\rho \leq \widetilde I_{\max}^{\ve/3}(A:B)_\rho + \log(1+72/\ve^2).
	\end{align}
	To show the result as Eq.~\eqref{appen eq:I max lemma}, we only need to prove 
	\begin{align}\label{appendix I max two}
		\widetilde I_{\max}^{\ve}(A:B)_\rho \leq I_{\max}^{\ve/2}(A:B)_\rho.
	\end{align}
	Denote $I_{\max}^{\ve/2}(A:B)_\rho = \lambda$ and suppose the optimal solution is taken at $\overline \rho_{AB}$ and $\sigma_B$. Let 
	\begin{align}
	\widetilde \rho_{AB} = \rho_A^{\frac12} V_A \overline \rho_A^{-\frac12} \overline\rho_{AB} \overline \rho_A^{-\frac12} V_A^\dagger \rho_A^{\frac12},
	\end{align}
	where $V_A$ is the unitary such that $F(\rho_A,\overline \rho_A) = \tr \rho_A^{\frac12} \overline \rho_A^{\frac12} V_A$.
By direct calculation, we have
\begin{align}
\widetilde \rho_{AB} = \rho_A^{\frac12} V_A \overline \rho_A^{-\frac12} \overline\rho_{AB} \overline \rho_A^{-\frac12} V_A^\dagger \rho_A^{\frac12} \leq 2^{\lambda} \cdot \rho_A^{\frac12} V_A P_{\overline \rho_A} V_A^\dagger \rho_A^{\frac12} \ox \sigma_B \leq 2^{\lambda} \rho_A \ox \sigma_B,
\end{align}
where $P_{\overline \rho_A}$ is the projector on the support of $\overline \rho_A$. Then Eq.~\eqref{appendix I max two} follows as soon as we show $\widetilde \rho_{AB} \approx^\ve \rho_{AB}$. This is actually shown in the final steps of the proof for Theorem 2 in~\cite{Anshu2018a}. We repeat these steps here for completeness. 
Deonte the purification of $\bar \rho_{AB}$ as $\ket{\bar \rho_{ABC}}= \bar \rho_A^{\frac12} \ket{\Phi}_{A:BC}$ where $\ket{\Phi}_{A:BC}$ denotes the non-normalized maximally entangled pure state in the cut $A:BC$. Then $\ket{\widetilde \rho_{ABC}} = \rho_A^{\frac12} V_A \overline \rho_A^{-\frac12} \ket{\bar \rho_{ABC}}$ is a purification of $\widetilde \rho_{AB}$. We have 
\begin{align}
     F^2(\widetilde \rho_{AB},\bar \rho_{AB}) & \geq F^2(\widetilde \rho_{ABC},\bar \rho_{ABC})  = |\<\bar \rho_{ABC}|\widetilde \rho_{ABC}\>|^2\notag\\ 
     &\quad = \big|\<{\Phi}_{A:BC}\big|\bar \rho_A^{\frac12} \rho_A^{\frac12} V_A \overline \rho_A^{-\frac12} \bar \rho_A^{\frac12} \ket{\Phi}_{A:BC}\>\big|^2
      = \big|\tr \bar \rho_A^{\frac12} \rho_A^{\frac12} V_A\big|^2 = F^2(\rho_A,\bar\rho_A),
 \end{align}
 which implies $P(\widetilde \rho_{AB},\overline \rho_{AB}) \leq P(\rho_A, \overline \rho_A)$. Thus it holds
\begin{align}
	P(\widetilde \rho_{AB},\rho_{AB}) & \leq P(\widetilde \rho_{AB}, \overline \rho_{AB}) + P(\overline \rho_{AB}, \rho_{AB})\notag\\ 
    & \quad \leq P(\rho_A, \overline \rho_A) + P(\overline \rho_{AB}, \rho_{AB}) \leq 2 P(\overline \rho_{AB}, \rho_{AB}) \leq \ve,
\end{align}
which completes the proof.
\end{proof}

\begin{lemma}\label{two sets same}
For any pure state $\phi_{AA'}$ and quantum state $\rho_{AB}$ such that $\phi_A = \rho_A$, the following two sets are the same,
\begin{align}
\left\{\cN_{A'\to B}(\phi_{AA'}) \approx^{\ve} \rho_{AB}\ |\ \cN \in \ \rm{CPTP}(A':B)\right\} = \left\{\sigma_{AB} \approx^{\ve} \rho_{AB}\ |\ \sigma_{A} = \rho_A\right\}.
\end{align}
\end{lemma}

\begin{proof}
Denote the L.H.S and R.H.S as $\cS_1$ and $\cS_2$ respectively. It is clear that $\cS_1 \subseteq \cS_2$ and we now show the other direction. For any quantum state $\sigma_{AB} \in \cS_2$, denote $\overline \sigma_{AB} = \sigma_A^{-1/2}\sigma_{AB} \sigma_A^{-1/2}$. Then, we have $\overline \sigma_{AB} \geq 0$ and $\overline \sigma_{A} = \1_A$. From the Choi-Jamio\l{}kowski isomorphism, we know that there exists a CPTP map $\cN_{A'\to B}$ such that $\overline \sigma_{AB} = \cN_{A'\to B}(\Phi_{AA'})$, where $\Phi_{AA'}$ denotes the un-normalized maximally entangled state. Thus, we get $\sigma_{AB} = \cN_{A'\to B}(\sigma_A^{1/2}\Phi_{AA'} \sigma_A^{1/2})$. Denoting $\psi_{AA'} = \sigma_A^{1/2}\Phi_{AA'} \sigma_A^{1/2}$, we have that $\psi_{AA'}$ is a purification of $\sigma_A$ and since $\sigma_A = \rho_A = \phi_A$ we get that $\phi_{AA'}$ is also a purification of $\sigma_A$. Due to Uhlmann's theorem~\cite{Uhlmann1976}, there exists a unitary $U$ on the system $A'$ such that $\psi_{AA'} = \mathcal{U}(\phi_{AA'})$ with $\mathcal{U}(\cdot) = U \cdot U^\dagger$. Hence, we find $\sigma_{AB} = \cN \circ \mathcal{U} (\phi_{AA'})\in \cS_1$. This completes the proof.
\end{proof}

\begin{lemma} \label{two sets same 1}
For any quantum channel $\cN_{A'\to B}$ and $\ve \in (0,1)$, denote the optimal values 
     \begin{align}
         k_i := & \inf_{\sigma^n_{RAB} \in K_i} I_{\max}(RA:B)_{\sigma^n_{RAB}} \quad \text{with}\quad i = 1,2 \quad \text{and}\\
      K_1:=& \left\{\sigma^n_{RAB} \Big|\, \sigma^n_{RAB} \approx^{\ve} \cN_{A'\to B}^{\ox n}(\o_{RAA'}^n), \sigma^n_{RA} = \o^n_{RA} \right\},\\
       K_2:=& \left\{\sigma^n_{RAB} \Big|\, \sigma^n_{RAB} = \widetilde \cN_{A'\to B}^n (\o_{RAA'}^n) \approx^{\ve} \cN_{A'\to B}^{\ox n}(\o_{RAA'}^n), \,\widetilde \cN^n \in \text{\rm Perm}(A'^n:B^n)\right\}.
     \end{align}
    where $\o_{RAA'}^n$ is the purification of the de Finetti state $\o^n_{AA'} = \int \sigma_{AA'}^{\ox n} d(\sigma_{AA'})$ with pure states $\sigma_{AA'} = \ket{\sigma}\bra{\sigma}_{AA'}$ and $d(\cdot)$ the measure on the normalized pure states induced by the Haar measure. Then $k_1 = k_2$.
\end{lemma}
\begin{proof}
    It is clear that $K_2 \subseteq K_1$ thus $k_1 \leq k_2$. We need to show the opposite direction. In the following, let us consider $R^nA^n$ as the reference system. For any optimal quantum state $\sigma_{RAB}^n \in K_1$, according to Lemma~\ref{two sets same}, there exists a quantum channel $\widetilde \cN_{A'\to B}^n$ such that 
    \begin{align}
        \sigma_{RAB}^n = \widetilde \cN_{A'\to B}^n (\o_{RAA'}^n)\approx^{\ve} \cN_{A'\to B}^{\ox n}(\o_{RAA'}^n).\label{two set tmp 2}
    \end{align}
     Then for any permutation operation $\pi^n$, we have $\pi^n_{A'} (\o_{A'}^n) = \o_{A'}^n$. Then both $\pi^n_{A'} (\o_{RAA'}^n)$ and $\o_{RAA'}^n$ are purifications of $\o_{A'}^n$. By Uhlmann's theorem~\cite{Uhlmann1976}, there exists a unitary $U^n_{\pi,RA}$ acting on the reference system $R^n A^n$ such that $\pi^n_{A'} (\o_{RAA'}^n) = \mathcal U_{\pi,RA}^n (\o_{RAA'}^n)$ with $\mathcal U_{\pi,RA}^n(\cdot) = U^n_{\pi,RA} (\cdot) (U^n_{\pi,RA})^\dagger$. Then we have
    \begin{align}
        \pi^n_B \circ \widetilde \cN_{A'\to B}^n  \circ \pi^n_{A'} (\o_{RAA'}^n) = \pi^n_B \circ \widetilde \cN_{A'\to B}^n  \circ \mathcal U^n_{\pi,RA} (\o_{RAA'}^n) = \mathcal U^n_{\pi,RA}\circ \pi^n_B \circ \widetilde \cN_{A'\to B}^n   (\o_{RAA'}^n),\label{two set eq tmp 1}\\
        \pi^n_B \circ  \cN_{A'\to B}^{\ox n}  \circ \pi^n_{A'} (\o_{RAA'}^n) = \pi^n_B \circ \cN_{A'\to B}^{\ox n}  \circ \mathcal U^n_{\pi,RA} (\o_{RAA'}^n) = \mathcal U^n_{\pi,RA}\circ \pi^n_B \circ\cN_{A'\to B}^{\ox n}   (\o_{RAA'}^n).
    \end{align}
    Since the purified distance is invariant under unitary operations and $\widetilde \cN_{A'\to B}^n (\o_{RAA'}^n)\approx^{\ve} \cN_{A'\to B}^{\ox n}(\o_{RAA'}^n)$, we have 
    \begin{align}
        \pi^n_B \circ \widetilde \cN_{A'\to B}^n  \circ \pi^n_{A'} (\o_{RAA'}^n) \approx^{\ve} \pi^n_B \circ  \cN_{A'\to B}^{\ox n}  \circ \pi^n_{A'} (\o_{RAA'}^n) = \cN_{A'\to B}^{\ox n} (\o_{RAA'}^n),
    \end{align}
    where the equality follows from the permutation invariance of $\cN_{A'\to B}^{\ox n}$. Due to the convexity of the purified distance, we have
    \begin{align}
        \widetilde \sigma^n_{RAB}:=\frac{1}{n!} \sum_{\pi_n} \pi^n_B \circ \widetilde \cN_{A'\to B}^n  \circ \pi^n_{A'} (\o_{RAA'}^n) \approx^{\ve}\cN_{A'\to B}^{\ox n} (\o_{RAA'}^n).\label{two set eq tmp  3}
    \end{align}
    Note that $\frac{1}{n!} \sum_{\pi_n} \pi^n_B \circ \widetilde \cN_{A'\to B}^n  \circ \pi^n_{A'}$ is permutation invariant, which implies $\widetilde \sigma^n_{RAB} \in K_2$. 
    Then we have
\begin{align}
    k_2 & \leq I_{\max}(RA:B)_{\widetilde \sigma^n_{RAB}}\\
    & = I_{\max}(RA:B)_{\Theta(\widetilde \cN_{A'\to B}^n)}\\
    & \leq I_{\max}(RA:B)_{\widetilde \cN_{A'\to B}^n} \\
    & = I_{\max}(RA:B)_{\sigma_{RAB}^n}\\
    & = k_1,
\end{align}
where the first and second equalities follows from the fact that channel's max-information is independent on the input state (see Remark~\ref{rem: Imax input}), the second inequality follows by the monotonicity of the channel's max-information (see Remark~\ref{rem: Imax data-processing}) under the superchannel $\Theta(\cdot) =  \frac{1}{n!} \sum_{\pi_n} \pi^n_B (\cdot) \pi^n_{A'}$, the last equality follows from the optimality assumption of $\sigma_{RAB}^n$. This completes the proof.
\end{proof}

\begin{lemma}\label{continuity of mutual information}     
For any quantum states $\rho_{AB}$ and $\sigma_{AB}$ such that $\rho_A = \sigma_A$ and $\frac12\|\rho-\sigma\|_1 \leq \ve$, it holds that
\begin{align}
     |I(A:B)_{\rho}-I(A:B)_\sigma| \leq 2\ve \log |A| + (1+\ve)h_2\left(\frac{\ve}{1+\ve}\right),
 \end{align}
 where $h_2(\cdot)$ is the binary entropy.
\end{lemma}
\begin{proof}
Since $I(A:B)_\rho = H(A)_\rho - H(A|B)_\rho$, we have 
\begin{align}
|I(A:B)_{\rho}-I(A:B)_\sigma| = |H(A|B)_{\rho} - H(A|B)_{\sigma}|\leq 2\ve \log |A| + (1+\ve)h_2\left(\frac{\ve}{1+\ve}\right),
\end{align}
where $H(A)$ and $H(A|B)$ are von Neumann entropy and conditional entropy respectively. The second inequality follows from~\cite[Lemma 2.]{Winter2016}.
\end{proof}

\begin{lemma}\label{compose channel lemma}
	Suppose the effective channel $\widetilde \cM_{A_i\to B_o} = \Pi_{A_iB_i\to A_oB_o}\circ \cN_{A_o\to B_i}$ with quantum channel $\cN$ and bipartite no-signalling quantum supermap $\Pi$, then we have the relation in terms of the their \mbox{\Choi} matrices $J^{\widetilde{\cM}}_{A_iB_o} = \tr_{A_oB_i} [(J_{A_oB_i}^{\cN})^T\ox\1_{A_iB_o}]J_{A_iB_iA_oB_o}^{\Pi}$.
\end{lemma}
\begin{proof}
	This result is widely used when considering no-signalling codes, e.g.,~\cite{Leung2015c,Duan2016}. We include its proof here only for completeness. Since $\Pi_{A_iB_i\to A_oB_o}$ is required to be $B$ to $A$ no-signalling, we can write it as $\Pi_{A_iB_i\to A_oB_o} = \cD_{E_oB_i\to B_o} \cF_{E_i\to E_o} \cE_{A_i \to A_o E_i}$ with quantum operations $\cE,\cF,\cD$. Note that the inverse \Choi isomorphism is given by $\cN_{A\to B}(X_A) = \tr_A J^{\cN}_{AB}(X_A^T \ox \1_B)$. In the following we will not explicitly write out the identity $\1_B$ for simplicity. For any operator $X_{A_iB_i}$, we have 
    \begin{align}
        \Pi_{A_iB_i\to A_oB_o}(X_{A_iB_i}) & = \cD_{E_oB_i\to B_o} \cF_{E_i\to E_o} \cE_{A_i \to A_o E_i}(X_{A_iB_i})\\
        & = \cD_{E_oB_i\to B_o} \cF_{E_i\to E_o} \tr_{A_i} J^\cE_{A_iA_o E_i} X_{A_iB_i}^{T_{A_i}}\\
        & = \cD_{E_oB_i\to B_o} \tr_{E_i} J^{\cF}_{E_iE_o} \big[\tr_{A_i} (J^\cE_{A_iA_o E_i})^{T_{E_i}} X_{A_iB_i}^{T_{A_i}}\big]\\
        & = \tr_{E_oB_i} J^{\cD}_{E_oB_iB_o} \Big[\tr_{E_i} (J^{\cF}_{E_iE_o})^{T_{E_o}} \big[\tr_{A_i} (J^\cE_{A_iA_o E_i})^{T_{E_i}} X_{A_iB_i}^{T_{A_iB_i}}\big]\Big]\\
        & = \tr_{A_iB_i} \Big[\tr_{E_iE_o} J^{\cD}_{E_oB_iB_o}(J^{\cF}_{E_iE_o})^{T_{E_o}}(J^\cE_{A_iA_o E_i})^{T_{E_i}}\Big] X_{A_iB_i}^{T_{A_iB_i}}.
    \end{align}
    Thus we have 
    \begin{align}\label{channel compose tmp 1}
        J^{\Pi}_{A_iB_iA_oB_o} = \tr_{E_iE_o} J^{\cD}_{E_oB_iB_o}(J^{\cF}_{E_iE_o})^{T_{E_o}}(J^\cE_{A_iA_o E_i})^{T_{E_i}}.
    \end{align} 
    Repeating the above steps again for $\widetilde \cM_{A_i\to B_o}(X_{A_i}) = \cD_{E_oB_i\to B_o} \cF_{E_i\to E_o} \cE_{A_i \to A_o E_i} \cN_{A_o\to B_i} (X_{A_i})$, we have 
    \begin{align}
        J^{\widetilde M}_{A_iB_o} = \tr_{A_oB_iE_iE_o} J^{\cD}_{E_oB_iB_o}(J^{\cF}_{E_iE_o})^{T_{E_o}}(J^\cE_{A_iA_o E_i})^{T_{A_oE_i}} (J_{A_oB_i}^{\cN})^{T_{B_i}}
    \end{align}
    By Eq.~\eqref{channel compose tmp 1}, we have 
    \begin{align}
        J^{\widetilde M}_{A_iB_o} = \tr_{A_oB_i} (J^{\Pi}_{A_iB_iA_oB_o})^{T_{A_o}} (J_{A_oB_i}^{\cN})^{T_{B_i}} = \tr_{A_oB_i} (J_{A_oB_i}^{\cN})^{T_{A_oB_i}} J^{\Pi}_{A_iB_iA_oB_o},
    \end{align}
    which completes the proof.
\end{proof}

\end{appendices}

\end{document}